\documentclass[twocolumn]{autart}  

\usepackage{mypreamble}

\usepackage[T1]{fontenc}

\usepackage[skip=0pt, indent=15pt]{parskip}

\newcommand{\maketextstyle}{\textstyle} %
\newcommand{\thmspace}{\vspace{4pt plus 2pt minus 2pt}}

\expandafter\def\expandafter\normalsize\expandafter{%
    \normalsize
    \setlength\abovedisplayskip{4pt}
    \setlength\belowdisplayskip{4pt}
    \setlength\abovedisplayshortskip{4pt}
    \setlength\belowdisplayshortskip{4pt}
}

\begin{document}

\begin{frontmatter}

\title{On the Nyquist Criterion and the Scaled Relative Graph for SISO systems: A Generalized Circle Criterion}

\author[tue]{Julius P.~J. Krebbekx}\ead{j.p.j.krebbekx@tue.nl},
\author[tue,hungary]{Roland T\'oth}\ead{r.toth@tue.nl},
\author[tue]{Amritam Das}\ead{am.das@tue.nl}

\address[tue]{Control Systems group, Department of Electrical Engineering,  Eindhoven University of Technology, The Netherlands}                                          
\address[hungary]{Systems and Control Lab, HUN-REN Institute for Computer Science and Control, Budapest, Hungary}


\begin{abstract}
    \emph{Scaled Relative Graphs} (SRGs) provide a novel graphical frequency-domain method for the analysis of nonlinear systems. However, existing SRG-based methods are limited to stable open-loop dynamics and can lead to false conclusions when applied to unstable \emph{Linear and Time-Invariant} (LTI) components. In this paper, we identify the source of this pitfall and resolve it by introducing the extended SRG, which augments the SRG of an LTI system with the encirclement information of the Nyquist criterion. We show how the extended SRG can be used for stability and $H_\infty$-performance analysis of LTI systems and the computation of SRGs for unstable LTI systems. Building on these results, we derive graphical stability and incremental $L_2$-gain conditions for three canonical nonlinear feedback interconnections: the Lur'e system, the controlled Lur'e system, and an LTI plant with Lur'e controller. As a corollary, we derive a user-friendly generalized circle criterion that relates the $L_2$-gain directly to separation between the Nyquist curve and the critical disk from the classical circle criterion, requiring SRGs only for the proof, not its statement. The proposed approach is validated on several examples involving unstable LTI dynamics and nonlinearities, demonstrating applicability beyond the capabilities of  existing SRG methods.
\end{abstract}

\end{frontmatter}

\section{Introduction}

For \emph{Linear Time-Invariant} (LTI) systems, graphical system analysis using the Nyquist diagram~\cite{nyquistRegenerationTheory1932} is the cornerstone of control engineering. It is easy to use, allows for intuitive analysis and controller design methods and has been widely used in the industry. However, it is unclear how to systematically generalize graphical frequency domain methods to nonlinear system analysis and controller design. 

There have been various attempts in extending graphical analysis methods for nonlinear systems, but they are all either approximative, or limited in range of applicability. Classical results such as the circle criterion \cite{sandbergFrequencyDomainConditionStability1964}, which is based on the Nyquist concept, and the Popov criterion can predict the stability of a class of nonlinear systems. These methods are exact, but are not useful for performance shaping. Moreover, they are limited to Lur'e systems with sector bounded nonlinearities. An alternative approach is the describing function method \cite{krylovIntroductionNonlinearMechanics1947}, which is an approximate method based on the Nyquist diagram, which considers only the first Fourier coefficient of the amplitude-dependent frequency response. While various generalizations of the describing function exist~\cite{rijlaarsdamComparativeOverviewFrequency2017}, they are often tailored to specific input signals such as sinusoids. The \emph{Scaled Relative Graph} (SRG), first introduced in~\cite{ryuScaledRelativeGraphs2022}, was proposed as a new graphical method to analyze nonlinear feedback systems in~\cite{chaffeyGraphicalNonlinearSystem2023}. It is an exact method, and it is intuitive because of its close connection to the Nyquist diagram. Moreover, SRG analysis can provide performance bounds in terms of (incremental) $L_2$-gain. 

The effectiveness of SRGs has recently been shown in \cite{vandeneijndenScaledGraphsReset2024,krebbekxResetControllerAnalysis2025} for the analysis and design of reset controllers, in~\cite{quanScaledRelativeGraphs2025} for circuit analysis, and, in~\cite{baron-pradaImpactOperatingPoints2027} for power systems, to name a few. By leveraging the restriction of the SRG to specific input spaces, it is also possible to compute a nonlinear Bode diagram and define a bandwidth for nonlinear systems~\cite{krebbekxNonlinearBandwidthBode2025a,krebbekxAmplitudeDependentBodeDiagrams2026}, generalizing~\cite{pavlovFrequencyDomainPerformance2007}. %
Moreover, SRG methods have been developed for multivariable~\cite{degrootExploitingStructureMIMO2025,krebbekxMIMOPaperExtended2026,grootDissipativityFrameworkConstructing2025} and data-driven systems~\cite{nautaComputingScaledRelative2026}.

The SRG calculus from~\cite{ryuScaledRelativeGraphs2022} bounds the gain of a feedback interconnection on its domain, but cannot guarantee that this domain is all of $L_2$. In~\cite{chaffeyGraphicalNonlinearSystem2023,chaffeyHomotopyTheoremIncremental2025}, a homotopy argument is used to prove that the domain is all of $L_2$ and that the interconnection is well-posed. However, this homotopy argument relies on small-gain theory, and therefore restricts~\cite{chaffeyGraphicalNonlinearSystem2023} and the subsequent developments mentioned above to feedback systems with \emph{stable} loop components. This was also noted in~\cite{chenSoftHardScaled2026}, where the \emph{hard} SRG was introduced to address this limitation.%
The hard SRG is defined on $\Lte$ instead of $L_2$, and admits unstable loop components without a homotopy argument, analogous to~\cite{teelGraphsConicRelations1996}. Without the homotopy argument, however, well-posedness of a loop with unstable components is no longer guaranteed and must be assumed.

In this paper\footnote{The extended SRG and the stability theorem for Lur'e systems were published in the conference paper~\cite{krebbekxScaledRelativeGraph2025a}. This paper generalizes \cite{krebbekxScaledRelativeGraph2025a} to two other systems than the Lur'e system: the controlled Lur'e plant and the LTI plant with a Lur'e controller. We also develop the theory of the extended SRG, derive a different generalization of the circle criterion, and apply our methods to three new examples.}, we develop tools which allow the use of unstable LTI loop components in SRG analysis, while maintaining well-posedness guarantees. We show that simply dropping the homotopy argument to admit unstable LTI loop components creates a pitfall: SRG calculus can then certify stability of a closed loop that is unstable, contradicting the Nyquist criterion. We resolve this pitfall by adding the encirclement information of the Nyquist criterion to the SRG of the LTI operator, which yields the \emph{extended SRG} introduced in~\cite{krebbekxScaledRelativeGraph2025a}. For nonlinear feedback systems, the homotopy condition is replaced by requiring the SRG bound of the nonlinearity to be \emph{inflatable}, i.e., bounded and star-shaped about a real point $\kappa$, similar to star-shaped uncertainty sets in IQC analysis~\cite{schererDissipativityIntegralQuadratic2022}. The homotopy then starts from the loop in which the nonlinearity is replaced by $\kappa$, for which stability is guaranteed via the extended SRG. As in~\cite{chaffeyHomotopyTheoremIncremental2025}, this guarantees well-posedness. Our main contributions are as follows.
\begin{itemize}
    \item \emph{Extended SRG} (Section~\ref{sec:pitfall_and_resolution}): we identify the pitfall, resolve it with the extended SRG, and use the extended SRG to characterize stability, the $H_\infty$-norm, and to compute the SRG of unstable LTI systems.
    \item \emph{Canonical interconnections} (Section~\ref{sec:extended_srg_and_nl_systems}): we derive graphical conditions for well-posedness, stability, and (incremental) $L_2$-gain bounds of the Lur'e system, the controlled Lur'e plant, and the LTI plant with a Lur'e controller, which admit unstable plants and integrating controllers. %
    \item \emph{Generalized circle criterion} (Corollary~\ref{cor:generalized_circle}): we bound the $L_2$-gain of the Lur'e system in terms of the distance between the Nyquist curve and the critical disk, where SRGs are only used in the proof.
    \item \emph{Examples} (Section~\ref{sec:examples}): we apply our results to a PID-controlled pendulum, a PD-controlled Duffing oscillator, and a controlled Lur'e plant with saturation, each containing unstable dynamics.
\end{itemize}
Moreover, Section~\ref{sec:preliminaries} presents the preliminaries, and Section~\ref{sec:conclusion} concludes the paper.

\section{Preliminaries}\label{sec:preliminaries}

In this section we collect the essential mathematical details required for the developments in Sections~\ref{sec:pitfall_and_resolution} and~\ref{sec:extended_srg_and_nl_systems}. Readers that are experts on the topic of SRGs may skip this section. We also extend two existing results to our setting: Lemma~\ref{lemma:L2_norm_to_L2e} and Theorem~\ref{thm:chaffey_thm2}.

\subsection{Notation and Conventions} 

Let $\R, \C$ denote the real and complex number fields, respectively, with $\R_{\geq 0} = [0, \infty)$, $j$ being the imaginary unit, and $\bar{z}$ being the complex conjugate of $z \in \C$. Let $\C_\infty := \C \cup \{ \infty \}$ denote the extended complex plane, where $z + \infty = \infty$ for $z \in \C_\infty$, $0 \cdot \infty = \infty, 1/0=\infty, 1/\infty = 0$ as in~\cite{ryuScaledRelativeGraphs2022}. For non-empty\footnote{See~\cite{ryuScaledRelativeGraphs2022} how to include empty sets.} sets $A, B \subseteq \C_\infty$, the sum and product sets are defined as $A+B:= \{ a+b \mid a\in A, b\in B\}$ and $AB:= \{ ab \mid a\in A, b\in B\}$, respectively. A closed disk in $\C$ is denoted by $D_r(x) = \{ z \in \C \mid |z-x| \leq r \}$, where $D_\infty(x) := \C$, while $D_{[\alpha, \beta]}$ is the closed complex disk in $\C$ centered on $\R$, intersecting $\R$ in $[\alpha, \beta]$. The radius of a set $\mathcal{C} \subseteq \C_\infty$ is defined as $\rmin(\mathcal{C}) := \inf_{r>0} : \mathcal{C} \subseteq D_r(0)$ if $\mathcal{C} \subseteq \C$, else $\infty$. The distance between non-empty sets $\mathcal{C}_1,\mathcal{C}_2 \subseteq \C_\infty$ is defined as $\dist(\mathcal{C}_1,\mathcal{C}_2) := \inf_{z_1 \in \mathcal{C}_1, z_2 \in \mathcal{C}_2} |z_1-z_2|$, where $|\infty \pm \infty|:=\infty$. Rational functions in $s \in \C$ with real coefficients are denoted as $\R(s)$, and stable\footnote{In the sense that all poles $p$ of $R(s)$ obey $\operatorname{Re}(p)<0$.} real rational transfer functions are $RH_\infty$.

\subsection{Signals, Systems and Stability}\label{sec:signals_systems_stability}

For sets $X,Y$, a relation $R : X \rightrightarrows Y$ is a possibly multi-valued map, defined by $Rx \subseteq Y$ for all $x \in X$. A relation has domain $\dom(R) := \{x \in X \mid Rx \neq \emptyset \}$ and range $\ran(R) := \{ y\in Y \mid \exists x \in X : y \in Rx \}$. The graph of a relation $R$ is the set $\operatorname{gph}(R) :=\{ (x,y) \in X \times Y \mid  y \in R x \}$. As in~\cite{chaffeyGraphicalNonlinearSystem2023,ryuScaledRelativeGraphs2022}, we do not distinguish a relation from its graph. The inverse relation is defined as $R^{-1} = \{ (y,x) \in Y \times X \mid (x,y) \in R \}$. A single-valued relation is called an operator, where $R : \dom(R) \subseteq X \to Y$ means that each $x \in \dom(R)$ maps to a unique $Rx \in Y$. 

Since this work focuses on \emph{Single-Input Single-Output} (SISO) continuous-time systems, the Hilbert space of particular interest is $L_2:= \{ f:\R_{\geq 0} \to \mathbb{R} \mid \norm{f} < \infty \}$, where the norm is induced by the inner product $\inner{f}{g}:= \int_0^\infty f(t) g(t) d t$. For any $T \in \R_{\geq 0}$ and $u : \R_{\geq 0} \to \R$, the truncation operator $P_T$ is defined as $(P_T u)(t) = 0$ for all $t>T$, else $(P_T u)(t) = u(t)$. The extension of $L_2$, see~\cite[p. 38]{desoerFeedbackSystemsInputoutput1975}, is defined as 
\begin{equation*}
    \Lte := \{ u : \R_{\geq 0} \to \R \mid \norm{P_T u} < \infty, \,  \forall T \in \R_{\geq 0} \}.
\end{equation*}
This space is the natural setting for modeling systems, as it includes periodic signals, which are otherwise excluded from $L_2$. On the other hand, extended spaces are not even normed spaces~\cite[Ch. 2.3]{willemsAnalysisFeedbackSystems1971}. Therefore, the Hilbert space $L_2$ is the adequate signal space for rigorously analyzing feedback systems using functional analysis. 

As a convention, we use the subscript \textrm{e} to denote a relation on $\Lte$, e.g., $R_\mathrm{e} : \Lte \to \Lte$. The restriction of $R_\mathrm{e}$ to $L_2$ is defined as $R := R_\mathrm{e}|_{L_2 \times L_2} = R_\mathrm{e} \cap (L_2 \times L_2)$.

An operator $R : L_2 \to L_2$ is said to be causal on $L_2$ if it satisfies $P_T (Ru) = P_T(R(P_Tu))$ for all $u \in L_2$ and $T \in \R_{\geq 0}$, i.e., the output at time $t$ is independent of the signal at times greater than $t$. Causality on $\Lte$ similarly defined. Causal systems $R:L_2 \to L_2$ are extended to $\Lte$ by defining $R_\mathrm{e} : \Lte \to \Lte$ as $P_T R_\mathrm{e} u := P_T R P_T u$, which is well-defined since $P_T u \in L_2$ for all $u \in \Lte$~\cite[p. 15]{willemsAnalysisFeedbackSystems1971}. If $R :L_2 \to L_2$ and $R_\mathrm{e} : \Lte \to \Lte$, then $R$ is causal on $L_2$ if and only if $R_\mathrm{e}$ is causal on $\Lte$~\cite[Ch. 2.4]{willemsAnalysisFeedbackSystems1971}. We model systems as operators $R_\mathrm{e} : \Lte \to \Lte$. Unless specified otherwise, we assume causality and $R0=0$.

For a relation $R: L_2 \rightrightarrows L_2$, the \emph{$L_2$-gain} is defined as
\begin{equation}\label{eq:non_incremental_induced_norm}
    \maketextstyle \gamma(R) := \sup_{\substack{ (x, y) \in R \\ x \neq 0}} \frac{\norm{y}}{\norm{x}} \in [0,\infty],
\end{equation}
and the \emph{incremental $L_2$-gain} is defined as
\begin{equation}\label{eq:incremental_induced_norm}
    \maketextstyle \Gamma(R) := \sup_{\substack{(x_1, y_1),(x_2, y_2) \in R \\ x_1 \neq x_2 }} \frac{\norm{y_1-y_2}}{\norm{x_1-x_2}}  \in [0,\infty].
\end{equation}
A system $R$ is said to be $L_2$-stable if $R : L_2 \to L_2$, i.e. $Ru \in L_2$ for all $u\in L_2$. When $R:L_2 \to L_2$ and $\gamma(R) < \infty$ ($\Gamma(R) < \infty$), we call the system (incrementally) stable with (incremental) gain $\gamma(R)$ ($\Gamma(R)$).

As~\eqref{eq:non_incremental_induced_norm} and~\eqref{eq:incremental_induced_norm} require a norm, and $\Lte$ is not a normed space, the (incremental) gain of $R_\mathrm{e} : \Lte \to \Lte$ is computed using $L_2$ signals only. Therefore, if $R_\mathrm{e} $ is \emph{not} $L_2$-stable (e.g., the system $G(s)= \frac{1}{s-1}$), its $L_2$-gain $\gamma(R)$ could be finite, since it is computed using inputs in $\dom(R) \subsetneq L_2$ such that $R : \dom(R) \to L_2$.

When $R : L_2 \to L_2$ is causal, the (incremental) gain can be equivalently characterized on $\Lte$ as in~\cite{zamesInputoutputStabilityTimevarying1966}. The following result is an extension of~\cite[Prop. 1.2.3]{vanderschaftL2GainPassivityTechniques2017} to our definition of (incremental) $L_2$-gain in~\eqref{eq:non_incremental_induced_norm} and~\eqref{eq:incremental_induced_norm}.

\thmspace
\begin{lemma}\label{lemma:L2_norm_to_L2e}
    Let $R: L_2 \to L_2$ be a causal operator, then $R_\mathrm{e}: \Lte \to \Lte$ is causal and 
    \begin{subequations}
    \begin{align}
        \gamma(R) &= \maketextstyle \sup_{0 \neq u \in \Lte} \sup_{T \in \R_{\geq 0}} \frac{\norm{P_T R_\mathrm{e} u}}{\norm{P_T u}}, \label{eq:non_incremental_induced_norm_L2e} \\ 
        \Gamma(R) &= \maketextstyle \sup_{\substack{u_1, u_2 \in \Lte \\ u_1 \neq u_2}} \sup_{T \in \R_{\geq 0}} \frac{\norm{P_T R_\mathrm{e} u_1-P_T R_\mathrm{e} u_2}}{\norm{P_T u_1-P_T u_2}}. \label{eq:incremental_induced_norm_L2e}
    \end{align}  
    \end{subequations}
\end{lemma}

\begin{proof}%
    Suppose $\Gamma(R) < \infty$. For any $u_1,u_2 \in \Lte$ and $T>0$, we have $\norm{P_T R_\mathrm{e} u_1-P_T R_\mathrm{e} u_2} = \norm{P_T R P_T u_1-P_T R P_T u_2} \leq \norm{R P_T u_1-R P_T u_2} \leq \Gamma(R) \norm{P_T u_1-P_T u_2}$. This also shows that \\ $\sup_{T \in \R_{\geq 0}} \frac{\norm{P_T R_\mathrm{e} u_1-P_T R_\mathrm{e} u_2}}{\norm{P_T u_1-P_T u_2}} \leq \Gamma(R)$, as in~\cite[Prop. 1.2.3]{vanderschaftL2GainPassivityTechniques2017}. By definition~\eqref{eq:incremental_induced_norm}, for all $\epsilon>0$, there are $u_1,u_2 \in L_2$ such that $\Gamma(R) -\epsilon \leq  \frac{\norm{Ru_1 - R u_2}}{\norm{u_1-u_2}}$. These $u_1,u_2$ obey $\sup_{T \in \R_{\geq 0}} \frac{\norm{P_T Ru_1-P_T Ru_2}}{\norm{P_T u_1-P_T u_2}} \geq \lim_{T \to \infty} \frac{\norm{P_T Ru_1-P_T Ru_2}}{\norm{P_T u_1-P_T u_2}} \geq \Gamma(R) - \epsilon$, which proves~\eqref{eq:incremental_induced_norm_L2e}. If $\Gamma(R) = \infty$, then there is a sequence $(u_{1k},u_{2k})_{k \in \N}$ such that $\lim_{k \to \infty} \frac{\norm{ Ru_{1k}-Ru_{2k}}}{\norm{u_{1k}- u_{2k}}} = \infty$, which attains $\lim_{k \to \infty} \sup_{T \in \R_{\geq 0}} \frac{\norm{P_T Ru_{1k}-P_T Ru_{2k}}}{\norm{P_T u_{1k}-P_T u_{2k}}} = \infty$. The statement in~\eqref{eq:non_incremental_induced_norm_L2e} follows by setting $u_2=0$. 
\end{proof}

Lemma~\ref{lemma:L2_norm_to_L2e} shows that for causal systems with full domain on $L_2$, the (incremental) gain carries over to $\Lte$. This result allows us to analyze $R$ on $L_2$ and extend the result to $\Lte$ using causality.

\subsection{Scaled Relative Graphs}\label{sec:srg_definitions}

We now turn to the definition and properties of the SRG, as introduced by Ryu et al. in~\cite{ryuScaledRelativeGraphs2022} and follow closely the exposition of the SRG by Chaffey et al. in~\cite{chaffeyGraphicalNonlinearSystem2023}.

\subsubsection{Definitions}

Let $R : \mathcal{H} \rightrightarrows \mathcal{H}$ be a relation on a Hilbert space $\mathcal{H}$. The angle between nonzero $u, y\in \mathcal{H}$ is defined as 
\begin{equation}\label{eq:def_srg_angle}
    \maketextstyle \angle(u, y) := \cos^{-1} \frac{\mathrm{Re} \inner{u}{y}}{\norm{u} \norm{y}} \in [0, \pi],
\end{equation}
and $\angle(u,y) := 0$ if $u=0$ or $y=0$. Given $u_1, u_2 \in \mathcal{U} \subseteq \dom(R)$, we define the set of complex numbers
\begin{multline*}
    z_R(u_1, u_2) := \\ \maketextstyle \left\{ \frac{\norm{y_1-y_2}}{\norm{u_1-u_2}} e^{\pm j \angle(u_1-u_2, y_1-y_2)} \mid y_1 \in R u_1, y_2 \in R u_2 \right\} ,
\end{multline*}
if $u_1 \neq u_2$. When $u_1 = u_2$ and $Ru_1$ multi-valued, then $z_R(u_1,u_2) := \{ \infty \}$, else $z_R(u_1,u_2) := \emptyset$. The SRG of $R$ over the set $\mathcal{U} \subseteq \dom(R)$ is defined as
\begin{equation}\label{eq:srg_def}
    \maketextstyle \SRG_\mathcal{U} (R) := \bigcup_{u_1, u_2 \in \mathcal{U}} z_R(u_1, u_2) \subseteq \C_\infty,
\end{equation}
and $\SRG(R) := \SRG_{\dom(R)}(R)$. One can also define the \emph{Scaled Graph} (SG) at some particular input. The SG of a relation $R$ with one input $u^\star \in \dom(R)$ fixed and the other in set $\mathcal{U} \subseteq \dom(R)$ is defined as
\begin{equation}\label{eq:sg_def}
    \textstyle \SG_{\mathcal{U}, u^\star}(R) := \bigcup_{u \in \mathcal{U}} z_R(u, u^\star) ,
\end{equation}
and $\SG_{u^\star}(R) := \SG_{\dom(R), u^\star}(R)$. %

By definition of the SRG, the (incremental) $L_2$-gain of a relation $R:L_2 \rightrightarrows L_2$, defined in~\eqref{eq:non_incremental_induced_norm} (or \eqref{eq:incremental_induced_norm}), is equal to $\gamma(R) = \rmin(\SG_0(R))$ (or $\Gamma(R) = \rmin(\SRG(R))$).

\subsubsection{Operations on SRGs}\label{sec:operations_on_srgs}

Here, we discuss the SRG of sums, concatenations and inverses of relations, as developed in~\cite[Chapter 4]{ryuScaledRelativeGraphs2022}. 

Inversion of a point $z = re^{j \omega} \in \C$ is defined as the M\"obius inversion $r e^{j \omega} \mapsto (1/r)e^{j \omega}$. By $R^{-1}$ we mean the relational inverse. An operator $R$ satisfies the \emph{chord property} if for all $z \in \SRG(R) \setminus \{ \infty \}$ it holds that $[z, \bar{z}] \subseteq \SRG(R)$. An operator $R$ is said to satisfy the left-arc (right-arc) property if for all $z \in \SRG(R) \setminus \{ \infty \}$, it holds that $\operatorname{Arc}^-(z, \bar{z}) \subseteq \SRG(R)$ ($\operatorname{Arc}^+(z, \bar{z}) \subseteq \SRG(R)$). If $R$ satisfies the left-arc, right-arc, or both arc properties, it is said to satisfy \emph{an} arc property. See Appendix~\ref{app:complex_geometry} for the definition of chords and arcs.

\thmspace
\begin{proposition}[\hspace*{-3pt}{\cite[Thms.~4-7]{ryuScaledRelativeGraphs2022}}]\label{prop:srg_calculus}
    Let $0 \neq \alpha \in \R$ and let $R,S$ be relations on a Hilbert space $\mathcal{H}$. Then\footnote{Note that $\dom(R \alpha) = \frac{1}{\alpha} \dom(R)$, $\dom(R^{-1}) =\ran(R)$, $\dom(R+S) = \dom(R)\cap \dom(S)$ and $\dom(RS) \subseteq \dom(S)$.}, 
    \begin{enumerate}[label=\alph*.]
        \item\label{eq:srg_calculus_alpha} $\SRG(\alpha R) = \SRG(R \alpha) = \alpha \SRG(R)$,
        \item\label{eq:srg_calculus_plus_one} $\SRG(I + R) = 1 + \SRG(R)$, where $I$ denotes the identity on $\mathcal{H}$,
        \item\label{eq:srg_calculus_inverse} $\SRG(R^{-1}) = (\SRG(R))^{-1} =: \SRG(R)^{-1}$.
        \item\label{eq:srg_calculus_parallel} If at least one of $R, S$ satisfies the chord property, then $\SRG(R + S) \subseteq \SRG(R) + \SRG(S)$.
        \item\label{eq:srg_calculus_series} If at least one of $R, S$ satisfies an arc property, then $\SRG(R S) \subseteq \SRG(R) \SRG(S)$.
    \end{enumerate}
    If the SRGs above contain $\infty$ or are the empty set, extra care is required, see~\cite{ryuScaledRelativeGraphs2022}. 
\end{proposition}
\thmspace

We note that the proof of Proposition~\ref{prop:srg_calculus} in \cite{ryuScaledRelativeGraphs2022} trivially carries over to $\SG_0$ if all the relations $R,S$ satisfy $R0=S0=0$ (since then $0\in \dom(R)$ and $0\in \ran(R)$ is guaranteed). Under this condition, Proposition~\ref{prop:srg_calculus} can be restated with the substitution $\SRG \to \SG_0$, i.e. the SRG interconnection rules hold for the SG at zero.

\subsubsection{SRG Bounds for Common Operators}

Here, we present some results that yield a bound for the SRG of common operators, which are LTI operators and static nonlinear maps. Define the Nyquist diagram of a transfer function $G(s)$ as $\operatorname{Nyq}(G) := \{ G(j \omega)\mid \omega \in \R \}$.

\thmspace
\begin{theorem}[\hspace*{-3pt}{\cite[Thm. 4]{chaffeyGraphicalNonlinearSystem2023}}]\label{thm:lti_srg}
    Let $R : L_2 \to L_2$ be stable and LTI with transfer function $R(s)$, then $\SRG(R) \cap \C_{\mathrm{Im} \geq 0}$ is the hyperbolic convex hull (see Appendix~\ref{app:complex_geometry} for the definition) of $\operatorname{Nyq}(R) \cap \C_{\mathrm{Im} \geq 0}$.
\end{theorem}

\thmspace
\begin{proposition}[\hspace*{-3pt}{\cite[Prop. 9]{chaffeyGraphicalNonlinearSystem2023}}]\label{prop:static_nl_srg}
    If $\phi : L_2 \to L_2$\footnote{Defined by $u(t) \mapsto \phi(u(t))$, where $\phi : \R \to \R$ and $u \in L_2$.} satisfies the slope condition $k_1 \leq \frac{\phi(x) - \phi(y)}{x-y} \leq k_2$ for all $x,y\in \R$, denoted as $\partial \phi \in [k_1, k_2]$, then $\SRG(\phi) \subseteq D_{[k_1, k_2]}$.
\end{proposition} 

\thmspace
\begin{proposition}[\hspace*{-3pt}{\cite[Prop. 10]{chaffeyGraphicalNonlinearSystem2023}}]\label{prop:static_nl_srg_non_incremental}
    If $\phi : L_2 \to L_2$ satisfies the sector condition $k_1 \leq \frac{\phi(x)}{x} \leq k_2$ for all $x \in \R$, denoted as $\phi \in [k_1, k_2]$, then $\SG_{0}(\phi) \subseteq D_{[k_1, k_2]}$.
\end{proposition}
\thmspace

\subsection{Analysis of Feedback Systems via SRGs}\label{sec:nl_system_analysis}

In this section, we show how to analyze the stability of the feedback interconnection in Fig.~\ref{fig:chaffey_thm2} using the SRG. We extend the results from \cite{chaffeyGraphicalNonlinearSystem2023,chaffeyHomotopyTheoremIncremental2025} by adding a causality result and investigate equivalent relations.

\subsubsection{Stability Analysis using Scaled Relative Graphs}\label{sec:sys_analysis_w_srgs}

To analyze the stability of a feedback system, solutions must exist for all inputs. We use the following definition of \emph{well-posedness} of a feedback system~\cite{megretskiSystemAnalysisIntegral1997}.

\thmspace
\begin{definition}\label{def:well-posedness}
    The interconnection in Fig.~\ref{fig:chaffey_thm2}, where $H_{1 \mathrm{e}}, H_{2\mathrm{e}}  : \Lte \to \Lte$ are causal, is called well-posed if $I+ H_{2\mathrm{e}} H_{1\mathrm{e}}$ has a causal inverse on $\Lte$. 
\end{definition}
\thmspace

We prove the following theorem regarding incremental stability for any system $H_1$ feedback interconnected with $H_2$, as displayed in Fig.~\ref{fig:chaffey_thm2}. It first occurred in~\cite{chaffeyGraphicalNonlinearSystem2023} and was corrected in~\cite{chaffeyHomotopyTheoremIncremental2025}. We add the minor improvements of concluding \emph{causality} of the closed-loop, given that the subsystems $H_1$ and $H_2$ are causal, and equivalence of two different representations of the relation.

\thmspace
\begin{theorem}\label{thm:chaffey_thm2}
    Consider $H_1, H_2 : L_2 \to L_2$, where $\Gamma(H_1) < \infty$ and $\Gamma(H_2) < \infty$, and for all $\tau \in [0, 1]$
    \begin{equation}\label{eq:chaffey_thm2_sep}
        \dist(S_1^{-1}, -\tau S_2) \geq r_m >0,
    \end{equation}
    where $\SRG(H_i) \subseteq S_i \subseteq \C_\infty$ for $i\in \{1,2\}$, and at least one of $S_1,S_2$ obeys the chord property. Then, the relation
    \begin{equation}\label{eq:chaffey_closed_loop_representations}
        T = (H_1^{-1} + H_2)^{-1} = H_1 (I + H_2 H_1)^{-1}
    \end{equation}
    obeys $T : L_2 \to L_2$ and has incremental $L_2$-gain $\Gamma(T) \leq 1/r_m$~\cite{chaffeyHomotopyTheoremIncremental2025}. If both $H_1$ and $H_2$ are causal, $T_\mathrm{e}$ is well-posed.
\end{theorem}

\begin{proof}%
    First,~\eqref{eq:chaffey_closed_loop_representations} follows from defining $y=Tr$ and the computations $e=r-H_2 H_1 e \implies e=(I+H_2 H_1)^{-1} r$ and $y=H_1 e$ to obtain $T = H_1 (I+H_2 H_1)^{-1}$ and solving $H_1^{-1}y=r-H_2 y$ for $y$ to obtain $T=(H_1^{-1}+H_2)^{-1}$. As $H_1,H_2$ are operators, equivalence in terms of relations follows from Lemma~\ref{lemma:relational-equivalance}. Moreover,~\eqref{eq:chaffey_thm2_sep} implies that $0 \notin \SRG(H_1^{-1} + H_2)$, hence~\eqref{eq:chaffey_closed_loop_representations} is a single-valued relation, i.e., an operator. Uniqueness of $e$ and $y$, $\dom(T) = L_2$ and the gain bound follow from~\cite[Corollary 1]{chaffeyHomotopyTheoremIncremental2025}. Causality follows from applying~\cite[Theorem 2.11]{willemsAnalysisFeedbackSystems1971} for the subalgebra $\tilde{\mathcal{B}}^+$ of causal operators with finite incremental gain to each application of the Banach fixed point theorem in the proof of~\cite[Theorem 2]{chaffeyHomotopyTheoremIncremental2025}. 
\end{proof}

Equality in~\eqref{eq:chaffey_closed_loop_representations} holds in terms of relations, as is proven in Appendix~\ref{app:relations_feedback}. As gain is defined for relations, the (incremental) $L_2$-gain of the relation in~\eqref{eq:chaffey_closed_loop_representations} can be analyzed \emph{separately} from the well-posedness of that feedback interconnection.

Upon replacing $\SRG$ with $\SG_0$ in Theorem~\ref{thm:chaffey_thm2}, one obtains an $L_2$-gain bound~\cite{chaffeyHomotopyTheoremIncremental2025}. However, no conclusion can be made regarding the well-posedness and/or causality of the closed-loop system, and the SG cannot be used to establish single-valuedness of~\eqref{eq:chaffey_closed_loop_representations} either.

The requirement that $H_1$ is stable poses a severe limitation for the applicability of SRG methods to controller design, as it is impossible to design a stabilizing controller for an unstable plant using Theorem~\ref{thm:chaffey_thm2}. Stability of $H_1$ is required for the proof of Theorem~\ref{thm:chaffey_thm2}, as it uses a \emph{homotopy} argument; The homotopy condition~\eqref{eq:chaffey_thm2_sep} ensures that the loop is stable for all $\tau \in [0,1]$, transferring stability from $H_1$ to the closed feedback loop. Removing this limitation is the primary objective of this paper.

\begin{figure}[t]
\centering
\begin{subfigure}[t]{0.45\linewidth}
\centering

\tikzstyle{block} = [draw, rectangle, 
minimum height=1.5em, minimum width=2em]
\tikzstyle{sum} = [draw, circle, node distance={0.5cm and 0.5cm}]
\tikzstyle{input} = [coordinate]
\tikzstyle{output} = [coordinate]
\tikzstyle{pinstyle} = [pin edge={to-,thin,black}]

\begin{tikzpicture}[auto, node distance = {0.2cm and 0.5cm}]
    \node [input, name=input] {};
    \node [sum, right = of input] (sum) {$ $};
    \node [block, right = of sum] (lti) {$G(s)$};
    \node [coordinate, right = of lti] (z_intersection) {};
    \node [output, right = of z_intersection] (output) {}; %
    \node [block, below = of lti] (static_nl) {$\phi$};

    \draw [->] (input) -- node {$r$} (sum);
    \draw [->] (sum) -- node {$e$} (lti);
    \draw [->] (lti) -- node [name=z] {$y$} (output);
    \draw [->] (z) |- (static_nl);
    \draw [->] (static_nl) -| node[pos=0.99] {$-$} (sum);
\end{tikzpicture}

\caption{The Lur'e system, where $G$ is an LTI dynamic block and $\phi$ is a static nonlinearity.}
\label{fig:lure}
\end{subfigure}
\hfill
\begin{subfigure}[t]{0.45\linewidth}
    \centering

    \tikzstyle{block} = [draw, rectangle, 
    minimum height=1.5em, minimum width=2em]
    \tikzstyle{sum} = [draw, circle, node distance={0.5cm and 0.5cm}]
    \tikzstyle{input} = [coordinate]
    \tikzstyle{output} = [coordinate]
    \tikzstyle{pinstyle} = [pin edge={to-,thin,black}]
    
    \begin{tikzpicture}[auto, node distance = {0.2cm and 0.5cm}]
        \node [input, name=input] {};
        \node [sum, right = of input] (sum) {$ $};
        \node [block, right = of sum] (lti) {$H_1$};
        \node [coordinate, right = of lti] (z_intersection) {};
        \node [output, right = of z_intersection] (output) {}; %
        \node [block, below = of lti] (static_nl) {$H_2$};
    
        \draw [->] (input) -- node {$r$} (sum);
        \draw [->] (sum) -- node {$e$} (lti);
        \draw [->] (lti) -- node [name=z] {$y$} (output);
        \draw [->] (z) |- (static_nl);
        \draw [->] (static_nl) -| node[pos=0.99] {$-$} (sum);
    \end{tikzpicture}
    
    \caption{A general feedback interconnection where $H_1$ and $H_2$ can be LTI or NL static or dynamic blocks.}
    \label{fig:chaffey_thm2}
\end{subfigure}
\hfill
\begin{subfigure}[t]{\linewidth}
    \centering

    \tikzstyle{block} = [draw, rectangle, 
    minimum height=1.5em, minimum width=2em]
    \tikzstyle{sum} = [draw, circle, node distance={0.5cm and 0.5cm}]
    \tikzstyle{input} = [coordinate]
    \tikzstyle{output} = [coordinate]
    \tikzstyle{pinstyle} = [pin edge={to-,thin,black}]
    
    \begin{tikzpicture}[auto, node distance = {0.2cm and 0.5cm}]
        \node [input, name=input] {};
        \node [sum, right = of input] (sum) {$ $};
        \node [block, right = of sum] (lti) {$L(s)$};
        \node [coordinate, right = of lti] (z_intersection) {};
        \node [output, right = of z_intersection] (output) {}; %
        \node [coordinate, below = of lti] (static_nl) {};
    
        \draw [->] (input) -- node {$r$} (sum);
        \draw [->] (sum) -- node {$e$} (lti);
        \draw [->] (lti) -- node [name=z] {$y$} (output);
        \draw [-] (z) |- (static_nl);
        \draw [->] (static_nl) -| node[pos=0.99] {$-$} (sum);
    \end{tikzpicture}
    
    \caption{Feedback interconnection with LTI loop transfer $L(s)$.}
    \label{fig:linear_feedback}
\end{subfigure}
    \caption{Simple Feedback Systems.}
    \label{fig:simple_nl_systems}
\end{figure}
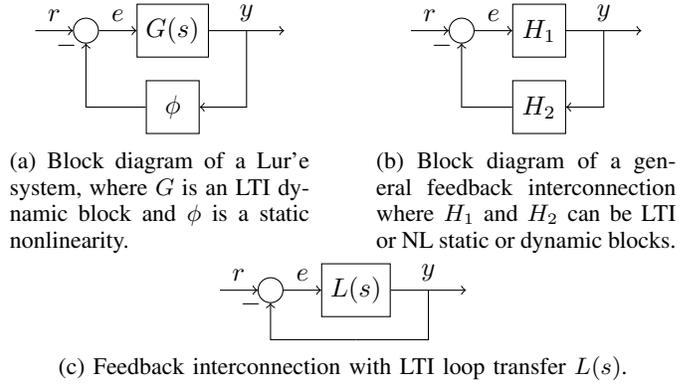

\subsubsection{Lur'e Systems}\label{sec:lure_systems}

An important special case of Fig.~\ref{fig:chaffey_thm2} is the Lur'e system. As depicted in Fig.~\ref{fig:lure}, the Lur'e system consists of an LTI block $G$ in feedback with a static nonlinear function $\phi : \R \to \R$. The closed-loop relation in~\eqref{eq:chaffey_closed_loop_representations} reads
\begin{equation}\label{eq:lure_closedloop}
    T = (G^{-1} + \phi)^{-1}.
\end{equation}
Stability of a Lur'e system with $\phi$ a sector bounded nonlinearity can be analyzed using the circle criterion, see Theorem~\ref{thm:circle}. Furthermore, in the special case that $\phi = \kappa \in \R$, stability of the LTI feedback loop can be analyzed using the Nyquist criterion, see Theorem~\ref{thm:nyquist}.

\section{Stability from SRGs: Pitfall and Resolution}\label{sec:pitfall_and_resolution}

The aim of this paper is to generalize Theorem~\ref{thm:chaffey_thm2} in a way to accommodate unstable loop components. We show why removing the homotopy condition from Theorem~\ref{thm:chaffey_thm2}, which is necessary to include unstable loop elements, results in a pitfall for stability analysis. In Section~\ref{sec:resolution}, we present one of our main results, which is a resolution to this pitfall by combining the Nyquist criterion with the SRG. This resolution opens up the possibility to use \emph{unstable} LTI systems in SRG computations.

\subsection{Pitfall of Stability Analysis with SRGs}\label{sec:pitfall}

When removing the homotopy condition from Theorem~\ref{thm:chaffey_thm2}, SRG analysis boils down to Proposition~\ref{prop:srg_calculus}. We will show that without the homotopy condition, SRG analysis contradicts the Nyquist criterion.

\subsubsection{An apparent contradiction}\label{sec:contraction}

Consider the feedback interconnection in Fig.~\ref{fig:linear_feedback}, where $L(s) = \frac{-2}{s^2+s+1}$.  %
First, we analyze the system using SRG calculus. Since $L(s)$ has poles $s=-1/2 \pm j \sqrt{3}/2$, it is stable, and $\SRG(L)$ is obtained via Theorem~\ref{thm:lti_srg}, see Fig.~\ref{fig:srg_L1}. The SRG of the closed-loop system, is obtained by first applying Proposition~\ref{prop:srg_calculus}.\ref{eq:srg_calculus_inverse} to obtain $\SRG(L^{-1})$ in Fig.~\ref{fig:srg_L1_inv}. Then, we use Proposition~\ref{prop:srg_calculus}.\ref{eq:srg_calculus_plus_one} to obtain $\SRG(1+L^{-1})$, see Fig.~\ref{fig:srg_L1_inv_p_1}, and finally we use Proposition~\ref{prop:srg_calculus}.\ref{eq:srg_calculus_inverse} again to obtain $\SRG(T)$, see Fig.~\ref{fig:srg_T1}. The radius of $\SRG(T)$, as obtained via SRG calculus, is clearly finite. This means that $T$ has finite incremental $L_2$-gain, which would show that $T$ is stable. 

\begin{figure*}[tb]
    \centering
    \begin{subfigure}[b]{0.14\linewidth}
        \centering
        \includegraphics[width=\linewidth]{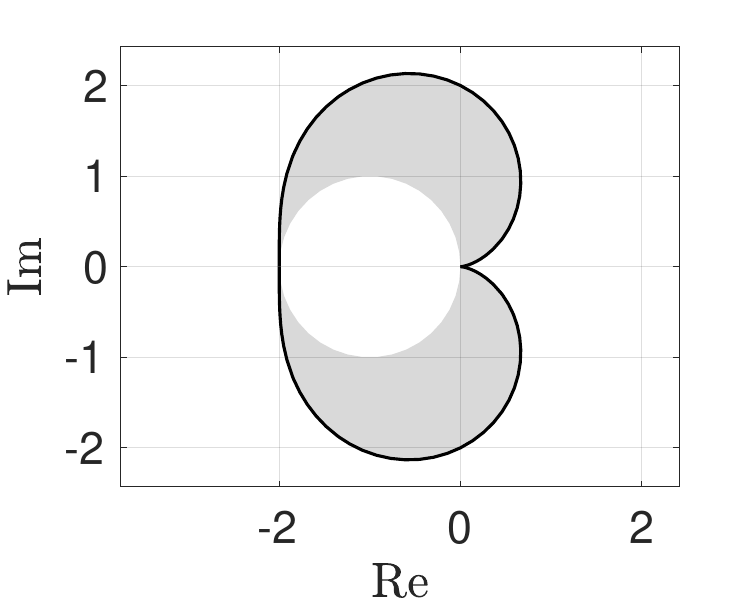}
        \caption{\tiny$\SRG(L)$.}
        \label{fig:srg_L1}
    \end{subfigure}
    \hfill
    \begin{subfigure}[b]{0.138\linewidth}
        \centering
        \includegraphics[width=\linewidth]{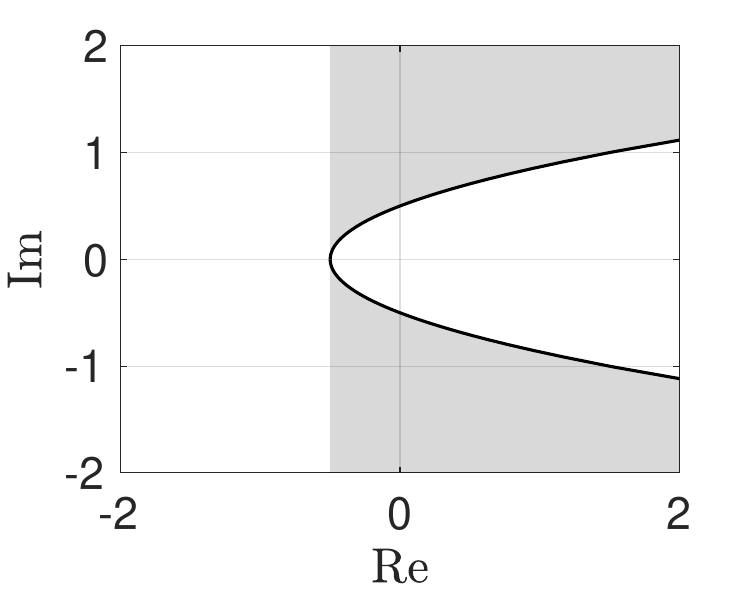}
        \caption{\tiny$\SRG(L)^{-1}$.}
        \label{fig:srg_L1_inv}
    \end{subfigure}
    \hfill
    \begin{subfigure}[b]{0.138\linewidth}
        \centering
        \includegraphics[width=\linewidth]{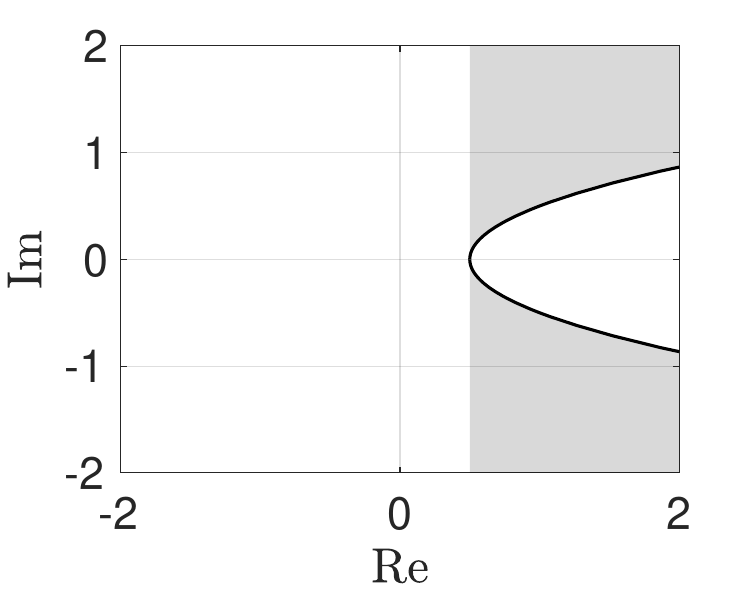}
        \caption{\tiny$1 + \SRG(L)^{-1}$.}
        \label{fig:srg_L1_inv_p_1}
    \end{subfigure}
    \hfill
    \begin{subfigure}[b]{0.14\linewidth}
        \centering
        \includegraphics[width=\linewidth]{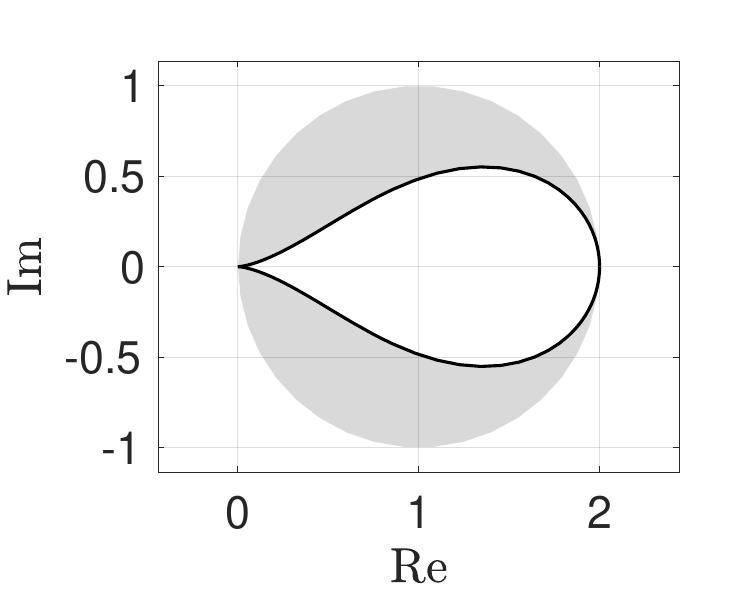}
        \caption{\tiny$\SRG(T)$.}
        \label{fig:srg_T1}
    \end{subfigure}
    \hfill
    \begin{subfigure}[b]{0.138\linewidth}
        \centering
        \includegraphics[width=\linewidth]{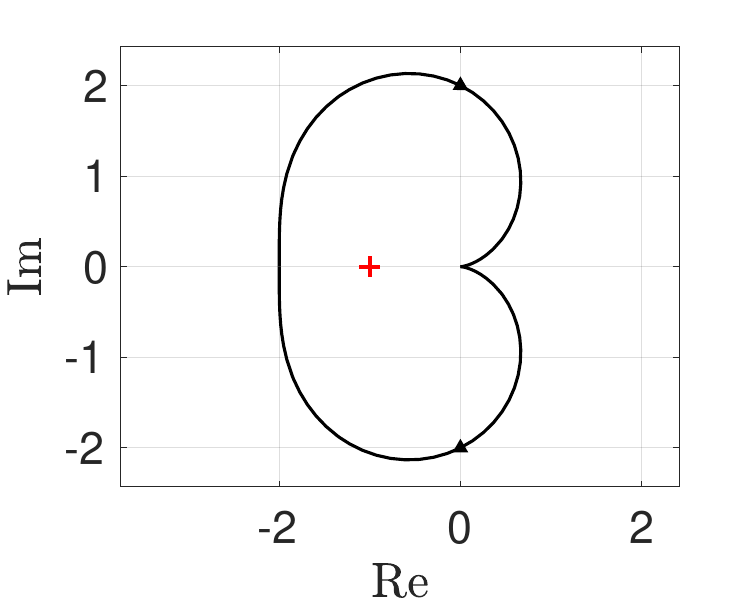}
        \caption{\tiny$\operatorname{Nyq}(L)$.}
        \label{fig:nyquist_L1}
    \end{subfigure}
    \hfill
    \begin{subfigure}[b]{0.14\linewidth}
        \centering
        \includegraphics[width=\linewidth]{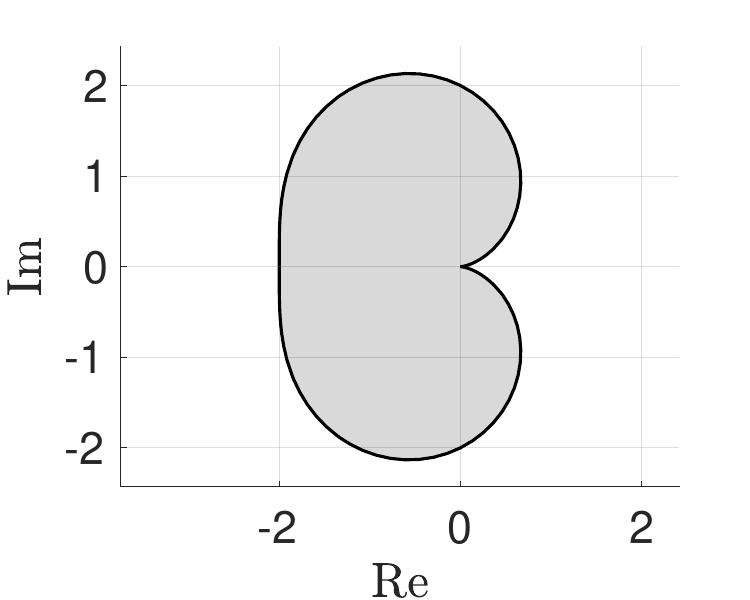}
        \caption{\tiny$\SRG'(L)$.}
        \label{fig:srg_L1_extended}
    \end{subfigure}
    \hfill
    \begin{subfigure}[b]{0.13\linewidth}
        \centering
        \includegraphics[width=\linewidth]{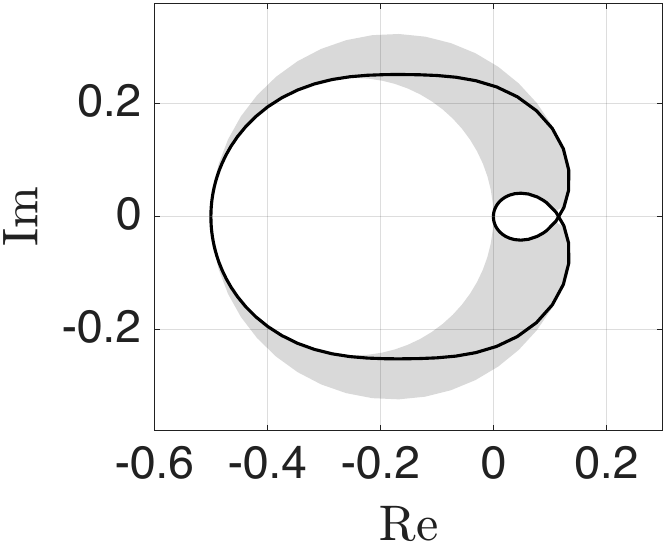}
        \caption{\tiny$\SRG(G_\mathrm{ex})$.}
        \label{fig:srg_Gex}
    \end{subfigure}
    \caption{SRGs and Nyquist diagram corresponding to $L(s) = \frac{-2}{s^2+s+1}$ and the example $G_\mathrm{ex}(s)$. The shaded area is the SRG and the bold line is the Nyquist curve.}
    \label{fig:srg_analysis_L1_T1}

\end{figure*}

However, when using the Nyquist criterion in Theorem~\ref{thm:nyquist} (see the appendix) to analyze stability, one concludes that $T(s)$ is unstable. This can be seen from the Nyquist diagram in Fig.~\ref{fig:nyquist_L1}, which encircles $-1$ one time in clockwise fashion, which results in $n_\mathrm{z} =1$ in Theorem~\ref{thm:nyquist}, indicating that $T(s)$ has one unstable pole.

It appears that we have encountered a \emph{contradiction}; Nyquist theory correctly predicts instability, while using the rules of SRG calculus we arrive at a wrong result. 

This apparent dichotomy is reminiscent of the Nyquist diagram of an unstable plant. When a proper LTI plant is continuously transformed from stable to unstable, the Nyquist diagram only achieves an infinite radius at the transition point between stable and unstable behavior. An example is $T_a(s) = 1/(s+a)$, which is stable for $a>0$, unstable for $a<0$, and achieves an infinite radius only at the transition point $a=0$. For LTI systems, the fact that the radius only attains infinity at a transition point is not a problem, since we can use the Nyquist criterion to assess stability. However, for SRGs, we do not have access to this information.

\subsubsection{Understanding the Pitfall}

The key to understanding this pitfall is to analyze the \emph{domain} and \emph{range} of $T : \dom(T) \subseteq L_2 \to L_2$ for which the SRG is computed. This is important since the SRG is defined for relations on $L_2$, not $\Lte$. Therefore, it can be true that $\dom(T) \neq L_2$, while $T_\mathrm{e} : \Lte \to \Lte$.

By Proposition~\ref{prop:srg_calculus}, it is indeed true that $\SRG(T) = (1+\SRG(L)^{-1})^{-1}$. However, since $T(s)$ has an unstable pole, one finds that for $u(t) = 1$ for $0 \leq t \leq 1$ and $u(t) = 0$ else,
it holds that $u \in L_2$ but $T_\mathrm{e} u \in \Lte \setminus L_2$ with $\lim_{\tau \to \infty} \norm{P_\tau T_\mathrm{e} u} = \infty$, hence $\dom(T) \neq L_2$. Therefore, Fig.~\ref{fig:srg_T1} correctly displays $\SRG_{\dom(T)}(T)$, but since $\dom(T) \neq L_2$, one cannot claim that $T$ is stable from $\Gamma(T) = \rmin(\SRG(T)) < \infty$, as the gain bound only holds on $\dom(T)$. To avoid this pitfall, one must establish $\dom(T) = L_2$ using other methods.

\subsubsection{The SRG of an Unstable LTI System}\label{sec:srg_of_unstable_lti}

For an unstable LTI plant, the Nyquist diagram, or the Bode diagram for that matter, can be interpreted as gain and phase per frequency when the unstable plant is part of an \emph{internally stable} feedback system. That is, the plant only receives signals that yield a stable behavior. Let $\mathcal{U}_\mathrm{e} \subseteq \Lte$ denote the signals that stabilize the unstable LTI plant $T(s)$, i.e. $\lim_{t \to \infty} \norm{P_t T_\mathrm{e} u}/\norm{P_t u} < \infty$ for all $u \in \mathcal{U}_\mathrm{e}$. For $T: \dom(T) \to L_2$, one has $\dom(T) = \mathcal{U}_\mathrm{e} \cap L_2 =: \mathcal{U}$, per definition, and $\SRG_\mathcal{U}(T) = \SRG(T)$, depicted in Fig.~\ref{fig:srg_T1}, which turns out to be exactly the h-convex hull of $\operatorname{Nyq}(T)$ as we will prove later on.

\subsection{Resolution of the Pitfall}\label{sec:resolution}

The pitfall does not come as a surprise, as the SRG disregards any information provided by the Nyquist criterion. As a resolution, we show that the Nyquist criterion can be combined with the SRG, such that direct application of SRG calculus leads to consistent results.

\subsubsection{Combining the Nyquist Criterion and the SRG}

We define the \emph{extended SRG} as follows, where the encirclement information from the Nyquist criterion is added to the SRG of an LTI system in Theorem~\ref{thm:lti_srg}.

\thmspace
\begin{definition}\label{def:lti_srg_extended}
    Let $G \in \R(s)$ be a proper LTI operator with $n_\mathrm{p}$ poles $p$ that obey $\mathrm{Re}(p) >0$, and define 
    \begin{equation}\label{eq:set_of_encircled_unstable_points}
        \mathcal{N}_G  = \{ z \in \C \setminus \operatorname{Nyq}(G) \mid N_G(z) +n_\mathrm{p} >0 \},
    \end{equation}
    where the winding number $N_G(z) \in \Z$ denotes the amount of clockwise encirclements of $z$ by $\operatorname{Nyq}(G)$\footnote{See Theorem~\ref{thm:nyquist} for details on the Nyquist contour.}. The extended SRG of an LTI operator $G$ is defined as 
    \begin{equation}\label{eq:lti_srg_redefinition}
        \SRG'(G) := \hco(\operatorname{Nyq}(G)) \cup \mathcal{N}_G,
    \end{equation}
    where $\hco(\cdot)$ denotes the hyperbolic convex hull (see Appendix~\ref{app:complex_geometry} for the definition).
\end{definition}
\thmspace

To illustrate Definition~\ref{def:lti_srg_extended}, we have plotted the extended SRG for $L(s) = \frac{-2}{s^2+s+1}$ in Fig.~\ref{fig:srg_L1_extended}. Because the \enquote{hole} in $\SRG(L)$ in Fig.~\ref{fig:srg_L1} is now filled in, one cannot derive the contradiction from Section~\ref{sec:contraction} anymore as $\SRG(T) \subseteq (1+\SRG'(L)^{-1})^{-1}$ and $0 \in 1+\SRG'(L)^{-1}$, hence we obtain no finite gain bound for $T$.

\subsubsection{SRG Analysis of LTI Systems}

Here, we study some basic properties of~\eqref{eq:lti_srg_redefinition}, and study its connection to the SRG and Nyquist diagram. In \cite{chaffeyGraphicalNonlinearSystem2023}, it is claimed that the SRG of an LTI operator is always bounded by its Nyquist diagram. Unfortunately, this is not true in general (e.g. for $G_\mathrm{ex}(s) = \frac{(s-5)(s+1)}{(s+4)(s+0.5)(s+5)}$, see Fig.~\ref{fig:srg_Gex}), but we can prove that the \emph{radius} is bounded by the Nyquist diagram.

\thmspace
\begin{lemma}\label{lemma:srg_radius_gamma_Gamma}
    Let $G \in RH_\infty$ be a stable LTI operator, then
    \begin{equation}
        \rmin(\SRG(G)) = \rmin(\operatorname{Nyq}(G)) = \Gamma(G) = \gamma(G).
    \end{equation}
\end{lemma}

\begin{proof}%
    Since $G$ is linear, one has $\rmin(\SRG(G))=\Gamma(G)=\gamma(G)$ by~\eqref{eq:srg_def}. %
    By Lemma~\ref{lemma:armin_abs}, for any $z_1,z_2 \in \C_{\mathrm{Im} \geq 0}$, one has $\rmin(\operatorname{Arc}_\mathrm{min}(z_1,z_2)) = \max \{|z_1|, |z_2|\}$, i.e., adding $\operatorname{Arc}_\mathrm{min}$ to $\operatorname{Nyq}(G)$ does not increase its radius, proving $\rmin(\SRG(G)) = \rmin(\operatorname{Nyq}(G))$.
\end{proof}

This lemma proves that the SRG (SG) radius corresponds to the $\Hinf$-norm. For a nonlinear operator, the $\Hinf$-norm is not defined. We only have the (incremental) gain in~\eqref{eq:non_incremental_induced_norm}  (or \eqref{eq:incremental_induced_norm}) to our disposal, which, at best, one can view as generalizations of the $\Hinf$-norm.

Lemma~\ref{lemma:srg_radius_gamma_Gamma} relates the SRG radius of \emph{stable} LTI systems to the $H_\infty$-norm. The next lemma combines stability analysis with Lemma~\ref{lemma:srg_radius_gamma_Gamma} by using the extended SRG.

\thmspace
\begin{lemma}\label{lemma:finite_srg'_radius}
    Let $G \in \R(s)$ be a proper LTI operator, then $G \in RH_\infty$ if and only if $\SRG'(G)$ is bounded. Moreover,
    \begin{equation}\label{eq:finite_srg'_radius}
        \rmin(\SRG'(G)) = \Gamma(G) = \gamma(G), \quad \forall G \in RH_\infty.
    \end{equation}
\end{lemma}
\begin{proof}%
    First, we prove $G \in RH_\infty \Leftrightarrow$ $\SRG'(G)$ is bounded. ``$\Leftarrow$'': as $\SRG'(G)$ is bounded, so is $\mathrm{Nyq}(G)$, hence all poles $p^G_i$ of $G$ must have nonzero real part and, as there are finitely many, obey $|\mathrm{Re}(p^G_i)| \geq \epsilon$ for some $\epsilon > 0$. Using the radius $r_G = \rmin(\SRG'(G))$, the Nyquist criterion implies that $T_k(s) := \frac{G(s)}{1+kG(s)}$ is stable for all $|1/k| > r_G$ where $\frac{1}{k} \in \R$. Therefore, all poles $p^{T_k}_i$ of $T_k$ must obey $\mathrm{Re}(p^{T_k}_i)<0$ for all $k \in (-\frac{1}{r_G},\frac{1}{r_G}) \setminus \{ 0 \}$. When $G(s) = b(s)/a(s)$, where $a,b$ are real polynomials in $s$, the poles of $T_k$ are the roots of $a(s) + k b(s)$. Since the roots of $a(s) + k b(s)$, i.e., the poles $p^{T_k}_i$, are continuous functions of $k$~\cite[App. D]{hornMatrixAnalysis1985}, it follows that $\lim_{k \to 0} \mathrm{Re}(p^{T_k}_i) < 0$ for each pole of $T_k$, proving that $T_0=G \in RH_\infty$. ``$\Rightarrow$'': as $G \in RH_\infty$, all poles $p^G_i$ of $G$ obey $\mathrm{Re}(p^G_i) \leq -\epsilon$ for some $\epsilon>0$. Again, using the fact that the roots of $a(s)+k b(s)$ are continuous in $k$, there must exist a $\delta >0$ such that for all $k \in (-\delta, \delta)$, the poles of $T_k$ obey $p^{T_k}_i \in (-2 \epsilon, 0)$. By the Nyquist criterion, this means that $\rmin(\mathcal{N}_G \cap \R) \leq 1/\delta$. Since $G \in RH_\infty \Rightarrow$ $\mathrm{Nyq}(G)$ bounded, it follows that $\SRG'(G)$ is bounded.

    To prove~\eqref{eq:finite_srg'_radius}, observe that when $G \in RH_\infty$, we have $\rmin(\SRG(G)) = \Gamma(G) = \gamma(G) \leq \rmin(\SRG'(G))$ by Lemma~\ref{lemma:srg_radius_gamma_Gamma}. Now we show that $\SRG'(G)$ and $\SRG(G)$ have the same radius. Since $\SRG(G) \subseteq \SRG'(G)$, one has $r':=\rmin(\SRG'(G)) \geq \rmin(\SRG(G))=:r$. Suppose that $r'> r$, then there must be $z_1 \in \SRG'(G) \setminus \SRG(G)$ and $z_2 \notin \SRG'(G)$ with $r<|z_1|<r'<|z_2|$ such that the winding numbers obey $N_G(z_1) \neq N_G(z_2)$. However, this is impossible since $\rmin(\operatorname{Nyq}(G)) = r$, which proves by contradiction that $r'=r$, proving~\eqref{eq:finite_srg'_radius}.
\end{proof}

\subsubsection{Computing the SRG of Unstable LTI Systems}

As suggested in Section~\ref{sec:srg_of_unstable_lti}, one could expect that the SRG of an unstable LTI system is closely related to its Nyquist diagram. In this section, we make this connection precise using the extended SRG.

\thmspace
\begin{theorem}\label{thm:srg_of_stabilizable_lti}
    Let $G \in \R(s)$ be proper. If $\R \setminus \SRG'(G) $ contains an open interval, then $\SRG(G) = \hco(\operatorname{Nyq}(G))$. 
\end{theorem}
\begin{proof}
    When $\R \setminus \SRG'(G) $ contains an open interval $S \subseteq \R$, then one can select a $0 \neq -1/k \in S$ such that $-1/k \notin \mathcal{N}_G$, see~\eqref{eq:set_of_encircled_unstable_points}. Fix such a $-1/k$, and define $T(s):=\frac{kG(s)}{1+kG(s)} = (1+(kG(s))^{-1})^{-1}$. By the Nyquist criterion, $-1/k \notin \mathcal{N}_G$ implies that $T \in RH_\infty$. 
    
    The relation $T$ can be inverted to yield $G = \frac{1}{k}(T^{-1}-I)^{-1}$, which implies $\SRG(G) = \frac{1}{k}(\SRG(T)^{-1}-1)^{-1}$ via Proposition~\ref{prop:srg_calculus}. Since $\mathrm{Nyq}(G) = \frac{1}{k}(\mathrm{Nyq}(T)^{-1}-1)^{-1}$, and $\SRG(T) = \hco(\operatorname{Nyq}(T))$ by Theorem~\ref{thm:lti_srg}, the result follows from Lemma~\ref{lemma:hco_w_inf} as all $z \in \SRG(G)$ lie on some $\mathrm{Arc}_\mathrm{min}(z_1,z_2)$ with $z_1,z_2 \in \mathrm{Nyq}(G)$. 
\end{proof}

Theorem~\ref{thm:srg_of_stabilizable_lti} provides a way to compute the SRG of $G : \dom(G) \subseteq L_2 \to L_2$ for those $G \in \R(s)$ that can be stabilized by a static feedback. For such a $G$, it provides an explanation for the observation in Section~\ref{sec:srg_of_unstable_lti}.

We note that it is not possible to derive similar interconnection rules as in Proposition~\ref{prop:srg_calculus} for the extended SRG. This can be readily seen from Proposition~\ref{prop:srg_calculus}.\ref{eq:srg_calculus_inverse}: $G^{-1}$ need not be proper, while the Nyquist criterion in Theorem~\ref{thm:nyquist} requires properness, which makes $\mathcal{N}_{G^{-1}}$ ill-defined. In the next section we show how the extended SRG can be used to generalize Theorem~\ref{thm:chaffey_thm2} to unstable LTI components and more general interconnections.

\section{The Extended SRG and Nonlinear Systems}\label{sec:extended_srg_and_nl_systems}

The extended SRG was used to derive some analysis results for LTI systems in Lemmas~\ref{lemma:srg_radius_gamma_Gamma} and~\ref{lemma:finite_srg'_radius}, and Theorem~\ref{thm:srg_of_stabilizable_lti}, however, its usefulness is still limited. In this section, we use the extended SRG from Definition~\ref{def:lti_srg_extended} to derive concrete formulas for the analysis of three canonical feedback interconnections with nonlinearities and \emph{unstable} LTI components. For Lur'e systems, we derive a generalization of the celebrated circle criterion.

\subsection{Canonical Systems}

We consider the Lur'e system in Fig.~\ref{fig:lure}, and the two common feedback interconnections in Figs.~\ref{fig:controlled_lure} and \ref{fig:lure_controlled_lti}.

\subsubsection{Controlled Lur'e plant}

The Lur'e plant controlled by a one degree of freedom LTI controller in feedback is displayed in Fig.~\ref{fig:controlled_lure}. If we write the complementary sensitivity operator as $T=(I+L^{-1})^{-1}$, then we can derive that $L = [(GK)^{-1} + K^{-1}\phi]^{-1}$ for the open-loop operator via $L = (G^{-1} + \phi)^{-1} K \iff L^{-1} = K^{-1} (G^{-1} + \phi) \iff L = ((GK)^{-1} + K^{-1} \phi)^{-1}$. The sensitivity operator $S : r \mapsto e$ is similarly obtained via $S = (I+L)^{-1}$.

\begin{figure}[tb]
    \centering

    \tikzstyle{block} = [draw, rectangle, 
    minimum height=1.5em, minimum width=2em]
    \tikzstyle{sum} = [draw, circle, node distance={0.5cm and 0.5cm}]
    \tikzstyle{input} = [coordinate]
    \tikzstyle{output} = [coordinate]
    \tikzstyle{pinstyle} = [pin edge={to-,thin,black}]
    
    \begin{tikzpicture}[auto, node distance = {0.2cm and 0.5cm}]
        \node [input, name=input] {};
        \node [sum, right = of input] (sum) {$ $};
        \node [block, right = of sum] (controller) {$K(s)$};
        \node [sum, right = of controller] (sigma) {$ $};
        \node [block, right = of sigma] (lti) {$G(s)$};
        \node [coordinate, right = of lti] (z_intersection) {};
        \node [output, right = of z_intersection] (output) {}; %
        \node [block, below = of lti] (static_nl) {$\phi$};
        \node [coordinate, right = of static_nl] (phi_intersection) {};
    
        \draw [->] (input) -- node {$r$} (sum);
        \draw [->] (sum) -- node {$e$} (controller);
        \draw [->] (controller) -- node {$u$} (sigma);
        \draw [->] (sigma) -- node {$u'$} (lti);
        \draw [->] (lti) -- node [name=z] {$y$} (output);
        \draw [->] (z) |- (static_nl);
        \draw [->] (static_nl) -| node[pos=0.99] {$-$} (sigma);
        \node [coordinate, below = of static_nl] (tmp1) {$H(s)$};
        \draw [->] (z) |- (tmp1)-| node[pos=0.99] {$-$} (sum);
    
    \end{tikzpicture}
    
    \caption{Block diagram of a Lur'e plant ($G,\phi$) controlled by the LTI controller $K$.}
    \label{fig:controlled_lure}
\end{figure}
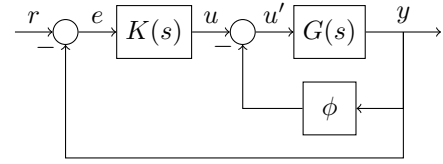

\subsubsection{LTI plant with Lur'e controller}

The LTI system with a Lur'e controller is depicted in Fig.~\ref{fig:lure_controlled_lti}. This system has complementary sensitivity operator $T = (I+L^{-1})^{-1}$ with $L = [(GK)^{-1} + \phi G^{-1}]^{-1}$, which is derived similarly as in the controlled Lur'e case.

\begin{figure}[tb]
    \centering

    \tikzstyle{block} = [draw, rectangle, 
    minimum height=1.5em, minimum width=2em]
    \tikzstyle{sum} = [draw, circle, node distance={0.5cm and 0.5cm}]
    \tikzstyle{input} = [coordinate]
    \tikzstyle{output} = [coordinate]
    \tikzstyle{pinstyle} = [pin edge={to-,thin,black}]
    
    \begin{tikzpicture}[auto, node distance = {0.2cm and 0.5cm}]
        \node [input, name=input] {};
        \node [sum, right = of input] (sum) {$ $};
        \node [sum, right = of sum] (sigma) {$ $};
        \node [block, right = of sigma] (lti_K) {$K(s)$};
        \node [coordinate, right = of lti_K] (u_intersection) {};
        \node [block, right = of u_intersection] (lti_G) {$G(s)$};
        \node [coordinate, right = of lti_G] (y_intersection) {};
        \node [output, right = of y_intersection] (output) {}; %
        \node [block, below = of lti_K] (static_nl) {$\phi$};
    
        \draw [->] (input) -- node {$r$} (sum);
        \draw [->] (sum) -- node {$e$} (sigma);
        \draw [->] (sigma) -- node {$e'$} (lti_K);
        \draw [->] (lti_K) -- node {$u$} (lti_G);
        \draw [->] (u_intersection) |- (static_nl);
        \draw [->] (static_nl) -| node[pos=0.99] {$-$} (sigma);
        \node [coordinate, below = of static_nl] (tmp1) {};
        \draw [->] (y_intersection) |- (tmp1)-| node[pos=0.99] {$-$} (sum);
        \draw [->] (lti_G) -- node {$y$} (output);
    
    \end{tikzpicture}
    
    \caption{Block diagram of an LTI system $G$ with Lur'e controller ($K, \phi$).}
    \label{fig:lure_controlled_lti}
\end{figure}
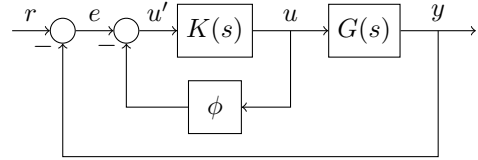

\subsection{Stability Analysis}

As in Theorem~\ref{thm:chaffey_thm2}, SRG stability analysis involves a \emph{bound} for the SRG of the operators involved, since the exact SRG is often unavailable. We require the following property regarding these bounds.

\thmspace
\begin{definition}\label{def:inflatable_set}
    A set $S \subseteq \C$ is called inflatable if it is bounded, and for some $\kappa \in S \cap \R$ and all $0 \leq \tau_1 \leq \tau_2 \leq 1$, it holds that $S_{\tau_1} \subseteq S_{\tau_2} $, where $S_\tau := \kappa + \tau(S-\kappa)$.
\end{definition}
\thmspace

Inflatability of $S$ is equivalent to $S-\kappa$ being a balanced or star-shaped set. The simplest inflatable sets are disks, which bound the SRG of slope-restricted and the SG of sector-bounded nonlinearities. More complex examples are the SG bounds of reset elements~\cite{vandeneijndenScaledGraphsReset2024} and the intersections of disks that bound the SRG of general bounded nonlinear systems~\cite{grootDissipativityFrameworkConstructing2025}.

The key property of inflatable sets for SRG analysis is that during the inflation $\tau \in [0,1)$, the set $S_\tau$ does not contain points that are not included in $S$ at $\tau = 1$. This allows one to replace $\kappa \to S$ directly, where $\SRG(\phi) \subseteq S$ for some desired operator $\phi$, without having to verify the homotopy condition in Theorem~\ref{thm:chaffey_thm2}. In this setting, $\kappa$ in Definition~\ref{def:inflatable_set} can be interpreted as a pre-stabilizing gain for robust/nonlinear stability analysis. In the context of stability analysis using Integral Quadratic Constraints (IQCs), this $\kappa$ corresponds to a stabilizing nominal static value of the uncertainty~\cite{schererDissipativityIntegralQuadratic2022}.

We now state stability theorems for the considered canonical interconnections. To make the result as practical as possible, we use the extended SRG for LTI operators in Definition~\ref{def:lti_srg_extended} to graphically check stability, and use inflatable bounds for the SRG of the nonlinearities to avoid the homotopy condition.

\thmspace
\begin{theorem}\label{thm:srg_nyquist_lure}
    For a proper LTI operator $G \in \mathbb{R}(s)$ and causal operator $\phi : L_2 \to L_2$, let $\Phi$ be an inflatable set such that $\SRG(\phi) \subseteq \Phi$. If either $\SRG(G)^{-1}$ or $\Phi$ satisfies the chord property, then 
    \begin{equation}\label{eq:thm_srg_nyquist_lure}
        \dist(\SRG'(G)^{-1}, -\Phi) = r >0
    \end{equation}
    implies that $T = (G^{-1}+\phi)^{-1}$ has $\dom(T) = L_2$ and $\Gamma(T) \leq 1/r$, and $T_\mathrm{e}$ is well-posed.
\end{theorem}
\begin{proof}
    Let $\kappa = \Phi_0 \in \Phi$ be the inflation center at $\tau = 0$ in Definition~\ref{def:inflatable_set}, which is not unique. Then,~\eqref{eq:thm_srg_nyquist_lure} implies that $-1/\kappa \notin \mathcal{N}_G$, hence $\tilde{G} = \frac{G}{1+\kappa G} \in RH_\infty$ thus $\Gamma(\tilde{G}) < \infty$ and $\tilde{G} :L_2 \to L_2$. Let $\tilde{\phi} = \phi-\kappa$, which corresponds to the loop transformation $T = (\tilde{G}^{-1} + \tilde{\phi})^{-1} = (G^{-1} + \phi)^{-1}$. Note that $\SRG(\tilde{\phi}) \subseteq \tilde{\Phi} := \Phi - \kappa$ holds since $\SRG(\tilde{\phi}) = \SRG(\phi)-\kappa$, and $\SRG(\tilde{G})^{-1} = \SRG(G)^{-1} + \kappa$. Then,~\eqref{eq:thm_srg_nyquist_lure} and inflatability of $\Phi$ implies $\dist(\SRG(\tilde{G})^{-1}, -\tau \tilde{\Phi}) \geq r$ for all $\tau \in [0,1]$. Well-posedness of $T_\mathrm{e}$, $\dom(T) = L_2$ and $\Gamma(T) \leq 1/r$ follow from Theorem~\ref{thm:chaffey_thm2} with $H_1 = \tilde{G}$ and $H_2 = \tilde{\phi}$.
\end{proof}

Theorem~\ref{thm:srg_nyquist_lure} was proven in~\cite{krebbekxScaledRelativeGraph2025a}. However, a different definition of well-posedness was used, and several mathematical details were missing. Here, we distinguish between operators on $L_2$ and $\Lte$, use a correct notion of domain for relations and operators, and use only proper LTI systems in $\R(s)$.

One could use Theorem~\ref{thm:srg_nyquist_lure} to analyze the controlled Lur'e system in Fig.~\ref{fig:controlled_lure} by multiplying the SRG bound for $(G^{-1}+\phi)^{-1}$ with $\SRG(K)$ to get a bound for the SRG of $L=(G^{-1} + \phi)^{-1} K $ and then using Theorem~\ref{thm:chaffey_thm2} with $H_1 = L, H_2=1$. However, this is not possible if $K \notin RH_\infty$ (e.g., a PID controller), or when the Lur'e system is not stable in open-loop. The next theorem gives an alternative stability criterion to cover these cases.

\thmspace
\begin{theorem}\label{thm:srg_nyquist_controlled_lure}
    Let $G,K \in \mathbb{R}(s)$ with $K^{-1} \in RH_\infty$, $\phi : L_2 \to L_2$ a causal operator, and $\Phi$ be an inflatable set such that $\SRG(\phi) \subseteq \Phi$. Let $L_0:= GK/(1+ \kappa G)$, $\tilde{\phi} := \phi-\kappa$, and $\tilde{\Phi} := \Phi-\kappa$, where $\kappa = \Phi_0$ is an inflation center for $\Phi$. If $L_0$ is proper and either $1+\SRG(L_0)^{-1}$ or $\SRG(K^{-1}) \tilde{\Phi}$ obeys the chord property, then 
    \begin{equation}\label{eq:thm_srg_nyquist_controlled_lure}
        \dist(1+\SRG'(L_0)^{-1}, -\SRG(K^{-1}) \tilde{\Phi})=r
    \end{equation}
    for some $r>0$ implies that $T = (I+(GK)^{-1} + K^{-1} \phi)^{-1}$ has $\dom(T) = L_2$ and $\Gamma(T) \leq 1/r$, and $T_\mathrm{e}$ is well-posed.
\end{theorem}
\begin{proof}
    Note that $(GK)^{-1} + K^{-1} \phi = L_0^{-1} +K^{-1} \tilde{\phi}$, hence $T = (1+L_0^{-1} + K^{-1} \tilde{\phi})^{-1}$, corresponding to Fig.~\ref{fig:chaffey_thm2} with $H_1 = T_\mathrm{LTI} := (1+L_0^{-1})^{-1}$ and $H_2 = K^{-1} \tilde{\phi}$. As $\Phi$ is inflatable and $K^{-1} \in RH_\infty$, $0 \in \SRG(K^{-1}) \tilde{\Phi}$, hence~\eqref{eq:thm_srg_nyquist_controlled_lure} implies $\dist(\SRG'(L_0)^{-1}, -1) \geq r$, which implies $T_\mathrm{LTI} \in RH_\infty$ by Theorem~\ref{thm:srg_nyquist_lure} with $G = L_0$ and $\phi=I$. Since $T_\mathrm{LTI}$ is stable, we have $\SRG(T_\mathrm{LTI}) = (1+\SRG(L_0)^{-1})^{-1} \subseteq (1+\SRG'(L_0)^{-1})^{-1}$. 

    Note $\tau \SRG(K^{-1}) \tilde{\Phi} \subseteq \SRG(K^{-1}) \tilde{\Phi}$ for all $\tau \in [0,1]$, as $\Phi$ is inflatable. Therefore,~\eqref{eq:thm_srg_nyquist_controlled_lure} implies $\dist(1+\SRG(L_0)^{-1},-\tau \SRG(K^{-1}) \tilde{\Phi}) \geq r$ for all $\tau \in [0,1]$. The latter satisfies~\eqref{eq:chaffey_thm2_sep} applied to $H_1 = T_\mathrm{LTI}, H_2 = K^{-1} \tilde{\phi}$. Finally, as $T_\mathrm{LTI}$ and $K^{-1} \tilde{\phi}$ are causal, Theorem~\ref{thm:chaffey_thm2} yields $\dom(T) = L_2$, $\Gamma(T) \leq 1/r$, and well-posedness of $T_\mathrm{e}$. 
\end{proof}

The analog of Theorem~\ref{thm:srg_nyquist_controlled_lure} for the LTI plant with Lur'e controller reads as follows.

\thmspace
\begin{theorem}\label{thm:srg_nyquist_lure_controller}
    Let $G,K \in \mathbb{R}(s)$ with $G^{-1} \in RH_\infty$, $\phi : L_2 \to L_2$ a causal operator, and $\Phi$ be an inflatable set such that $\SRG(\phi) \subseteq \Phi$. Let $L_0:= GK/(1+ \kappa K)$, $\tilde{\phi} := \phi-\kappa$, and $\tilde{\Phi} := \Phi-\kappa$, where $\kappa = \Phi_0$ is an inflation center for $\Phi$. If $L_0$ is proper and either $1+\SRG(L_0)^{-1}$ or $\tilde{\Phi} \SRG(G^{-1})$ obeys the chord property, then 
    \begin{equation}\label{eq:srg_nyquist_lure_controller}
        \dist(1+\SRG'(L_0)^{-1}, -\tilde{\Phi} \SRG(G^{-1}) )=r
    \end{equation}
    for some $r>0$ implies that $T = (I+(GK)^{-1} +  \phi G^{-1})^{-1}$ has $\dom(T) = L_2$ and $\Gamma(T) \leq 1/r$, and $T_\mathrm{e}$ is well-posed.
\end{theorem}
\begin{proof}
    The proof is obtained by replacing $K^{-1} \phi$ with $\phi G^{-1}$ in the proof of Theorem~\ref{thm:srg_nyquist_controlled_lure}.
\end{proof}

Since $\dist(\SRG'(G)^{-1}, -\Phi) =\dist(\SRG(\tilde{G})^{-1}, - \tilde{\Phi})$ in the proof of Theorem~\ref{thm:srg_nyquist_lure}, the choice of inflation center $\kappa \in \Phi \cap \R$ is arbitrary and does not affect the resulting gain bound. The resulting $r$ in Theorems~\ref{thm:srg_nyquist_controlled_lure}, and~\ref{thm:srg_nyquist_lure_controller} may depend on the choice of inflation center $\kappa \in \Phi \cap \R$, which is not unique. Here, $\kappa$ plays the role of a loop transformation. Maximizing $r$ over all possible values $\kappa \in \R$ can be viewed as optimizing a static multiplier. The non-incremental variant of these theorems can be trivially derived using the SG instead of the SRG. In this case, well-posedness must be assumed instead, see~\cite{chaffeyHomotopyTheoremIncremental2025,vandeneijndenScaledGraphsReset2024}.

\subsection{A Generalized Circle Criterion}

The circle criterion is a graphical test for absolute stability~\cite[Thm. 7.1]{khalilNonlinearControl2002} or $L_2$-stability~\cite[Thm. V.2.10]{desoerFeedbackSystemsInputoutput1975} of a Lur'e system with sector bounded nonlinearity. Using the results from Section~\ref{sec:resolution}, we can generalize this classical result to include an $L_2$-gain bound, which is computed using only the Nyquist curve and the critical disk from the circle criterion. Even though the SRG is used in the proof, this generalization of the circle criterion can be stated and used without involving the SRG.

\thmspace
\begin{corollary}\label{cor:generalized_circle}
    Let $G \in \R(s)$ be proper, $\phi \in [k_1, k_2]$, and the conditions for Theorem~\ref{thm:circle} be satisfied. Then, the Lur'e systems $T_1 =(G^{-1} +\phi)^{-1}$ and $T_2 =(\phi^{-1} + G)^{-1}$ obey $T_i : L_2 \to L_2$ and $\gamma(T_i) \leq 1/r_i$ for $i \in \{1,2\}$, where
    \begin{subequations}\label{eq:circle_gains}
    \begin{align}
        r_1 &= \dist(\operatorname{Nyq}(G)^{-1},-D_{[k_1,k_2]}), \label{eq:circle_gain_lure} \\ 
        r_2 &= \dist(\operatorname{Nyq}(G), -(D_{[k_1,k_2]})^{-1}). \label{eq:circle_gain_lure_flipped}
    \end{align}
    \end{subequations}
\end{corollary}
\begin{proof}
    By Theorem~\ref{thm:circle}, we obtain $T_1 : L_2 \to L_2$ and well-posedness of $T_{1 \mathrm{e}}$ by setting $u_2=0$, and similarly for $T_{2\mathrm{e}}$ when $u_1=0$. It remains to prove $\gamma(T_i) \leq 1/r_i$ with $r_i$ from~\eqref{eq:circle_gains}. Since $\dom(T_i) = L_2$, we may conclude $\gamma(T_i) = \rmin(\SG_0(T_i))$ from~\eqref{eq:sg_def}. From the conditions of Theorem~\ref{thm:circle}, we know that $G$ can be stabilized by any static gain $k \in [k_1,k_2]$, so we obtain $\SRG(G) = \mathcal{G}_G$ from Theorem~\ref{thm:srg_of_stabilizable_lti}, and $\SG_0(\phi) \subseteq D_{[k_1,k_2]}$ from Proposition~\ref{prop:static_nl_srg_non_incremental}. 
    
    Next, we prove~\eqref{eq:circle_gain_lure}. By Proposition~\ref{prop:srg_calculus}, we have $\SG_0(T_1) \subseteq (\SRG(G)^{-1} + D_{[k_1,k_2]})^{-1}$ since $D_{[k_1,k_2]}$ satisfies the chord property, and $\SG_0(G) = \SRG(G)$ as $G$ is linear. Therefore, $\gamma(T_1) \leq 1/\tilde{r}_1$ where $\tilde r_1 := \dist(\SRG(G)^{-1},-D_{[k_1,k_2]})$. Finally,~\eqref{eq:circle_gain_lure} follows from showing $\tilde{r}_1=r_1$. For this step, note that $\SRG(G)^{-1}$ is the h-convex hull of $\mathrm{Nyq}(G)^{-1}$ by Lemma~\ref{lemma:hco_w_inf}, and that $\tilde{r}_1 = \dist(\SRG(G)^{-1}+c,D_{[-r,r]})$, where $c = \frac{k_1+k_2}{2}, r=\frac{k_2-k_1}{2}$. Since $\SRG(G)^{-1}$ consists of $\mathrm{Arc}_\mathrm{min}(z_1,z_2)$ with $z_1,z_2 \in \mathrm{Nyq}(G)^{-1}$, the closest distance to $D_{[-r,r]}$ in $\SRG(G)^{-1}+c$ is attained by some $z \in \mathrm{Nyq}(G)^{-1}$ according to Lemma~\ref{lemma:armin_abs}, proving $\tilde{r}_1=r_1$. 
    
    To prove~\eqref{eq:circle_gain_lure_flipped}, observe that $(D_{[k_1,k_2]})^{-1}$ is either: 
    \begin{itemize}
        \item [1.] a disk $D_{[1/k_1,1/k_2]}$ if $0 < k_1 < k_2$ or $k_1 < k_2 < 0$.
        \item [2.] a half-plane $\C_{\mathrm{Re} \geq 1/k_2} \cup \{ \infty \}$ if $k_1=0$ or $\C_{\mathrm{Re} \leq 1/k_1} \cup \{ \infty \}$ if $k_2=0$.
        \item [3.] a disk exterior $\C_\infty \setminus D_{[1/k_1,1/k_2]}$ if $k_1 < 0< k_2$. 
    \end{itemize}
    In Cases 1 and 2, $(D_{[k_1,k_2]})^{-1}$ satisfies the chord property, hence $\SG_0(T_2) \subseteq (\SRG(G) + (D_{[k_1,k_2]})^{-1})^{-1}$, which implies $\gamma(T_2) \leq 1/\tilde{r}_2$ with $\tilde{r}_2 := \dist(\SRG(G), -(D_{[k_1,k_2]})^{-1})$. In Case 3, we must add chords to $\SRG(G)$ such that it satisfies the chord property, resulting in the set $\overline{\SRG(G)}$, in which case we obtain $\SG_0(T_2) \subseteq (\overline{\SRG(G)} + (D_{[k_1,k_2]})^{-1})^{-1}$ and $\gamma(T_2) \leq 1/\tilde{r}_2$ with $\tilde{r}_2 := \dist(\overline{\SRG(G)}, -(D_{[k_1,k_2]})^{-1})$.

    Finally,~\eqref{eq:circle_gain_lure_flipped} follows from showing $\tilde{r}_2=r_2$. In Case 1, the same centering argument as for $\tilde{r}_1=r_1$, together with Lemma~\ref{lemma:armin_abs} for the infimum, shows $\tilde{r}_2=r_2$. Case 3 follows analogously, if Lemma~\ref{lemma:armin_abs} for the supremum is used together with the fact that $\dist([z,\bar{z}], -(D_{[k_1,k_2]})^{-1}) = \dist(z, -(D_{[k_1,k_2]})^{-1})$ for $z \in \overline{\SRG(G)}$, i.e., adding chords does not affect the distance to the disk interior containing $\SRG(G)$. In Case 2, Lemma~\ref{lemma:armin_re} shows that the smallest distance is attained by $\mathrm{Nyq}(G)$, hence proving $\tilde{r}_2=r_2$ in all cases.
\end{proof}

Multiple generalizations of the circle criterion already exist. For multivariable systems, one has criteria involving robustness using dynamic cones~\cite{safonovMultiloopGeneralizationCircle1981}, generalizations using IQCs~\cite{megretskiNecessarySufficientConditions1993}, and classes of dynamic nonlinearities~\cite{guiverCircleCriterionClass2022} while providing necessary and sufficient conditions for stability~\cite{guiverCircleCriterionClass2022,megretskiNecessarySufficientConditions1993}. Focusing more on performance, there are generalizations involving input-to-state stability, e.g., see~\cite{jayawardhanaCircleCriterionInputtoState2011}. To the best of the authors' knowledge, no existing generalization establishes a direct relationship between the $L_2$-gain, and the separation between the Nyquist curve and the critical disk.

The derivation of the incremental version of Corollary~\ref{cor:generalized_circle} is entirely analogous, though we do not pursue it here for brevity. Since the non-incremental circle criterion in Theorem~\ref{thm:circle} is a direct consequence of the non-incremental small-gain theorem~\cite[Thm. III.2.1]{desoerFeedbackSystemsInputoutput1975}, one derives the incremental circle criterion using the incremental small-gain theorem~\cite[Thm. III.3.1]{desoerFeedbackSystemsInputoutput1975}. The incremental version of Corollary~\ref{cor:generalized_circle} guarantees well-posedness, instead of assuming it, and provides incremental gain bounds $\Gamma(T_i) \leq 1/r_i$ with $r_i$ from~\eqref{eq:circle_gains} using a slope condition $\partial \phi \in [k_1,k_2]$ on $\phi$.

\section{Examples}\label{sec:examples}

We present three examples showcasing the use of the tools from Section~\ref{sec:extended_srg_and_nl_systems}. The first example, the nonlinear pendulum, demonstrates how Theorem~\ref{thm:srg_nyquist_controlled_lure} for controlled Lur'e systems is used. The second example is the controlled Duffing oscillator, which uses Theorem~\ref{thm:srg_nyquist_lure}. This example also shows how SRG analysis with a restricted input space is used. The third example shows how Theorems~\ref{thm:chaffey_thm2} and~\ref{thm:srg_nyquist_lure} can be combined to analyze more general interconnections than those discussed in Section~\ref{sec:extended_srg_and_nl_systems}. We refer the reader to~\cite{krebbekxScaledRelativeGraph2025a} for another example\footnote{In~\cite{krebbekxScaledRelativeGraph2025a}, $\dist(\SRG'(G)^{-1}, -\SRG(\phi))=1/4$ is reported. However, this is a typo and it should be $1/3$.} involving Theorem~\ref{thm:srg_nyquist_lure}, which we do not repeat here.

\subsection{Nonlinear Pendulum}\label{sec:examples_nl_pendulum}

The nonlinear pendulum is described by the ordinary differential equation (ODE) 
\begin{equation}
    m l^2 \ddot{\theta} + C l \dot{\theta} + m g l \sin(\theta) = u, 
\end{equation}
where $u$ is the torque applied by the controller, and length $l$, mass $m$, gravity $g$, and friction $C$ are parameters. The pendulum is represented in Lur'e form by
\begin{equation*}
    G(s) = \frac{1}{m l^2 s^2 + C l s}, \quad \phi(\theta) = mgl \sin(\theta),
\end{equation*}
where in Fig.~\ref{fig:controlled_lure}, the output $y$ is the pendulum angle $\theta$. For the sake of simplicity, we let the parameters be $m=l=C=g=1$ for the remainder of this example. We consider a PID controller $K(s)=k_p + k_i /s + k_d s/(s/N+1)$ that controls the motion of the pendulum by applying a torque. We consider two different PID controller designs, defined as $K_1(s) = 2 + 1/s + s/(s/N+1)$ and $K_2(s) = 5 + 1/s + 2s/(s/N+1)$, where $N=10 < \infty$ ensures that the controller is implementable (i.e., proper).

\subsubsection{Stability and Performance Analysis with SRGs}\label{sec:pendulum_stability_analysis_srg}

Note that $G(s), K_i(s)$ are proper and $K_i^{-1} \in RH_\infty$ for $i \in \{1,2\}$. For $\phi(x) = \sin(x)$, which satisfies $\partial \phi \in [-1, 1]$, we have $\SRG(\phi) \subseteq D_{[-1, 1]} =: \Phi$ by Proposition~\ref{prop:static_nl_srg}. Observe that $\Phi$ satisfies the chord property, and is inflatable with $0=\kappa \in \Phi$, hence $\tilde{\Phi} = \Phi$ and we obtain that $L_0=G K_i$ is proper. Since $\tilde{\Phi}$ is a disk centered at the origin, it follows that $\SRG(K_i^{-1}) \tilde{\Phi}$ is also a disk, so it satisfies the chord property. Finally, as $\phi$ is causal, we satisfy all the conditions for Theorem~\ref{thm:srg_nyquist_controlled_lure}. 

It remains to verify the separation condition in~\eqref{eq:thm_srg_nyquist_controlled_lure},
which is visualized in Fig.~\ref{fig:pendulum_srgs} for both controllers. From Theorem~\ref{thm:srg_nyquist_controlled_lure}, we can conclude that $T_\mathrm{e} : \Lte \to \Lte$ is a well-posed operator with $\Gamma(T) \leq 1/r_1$ for $K_1$ and $\Gamma(T) \leq 1/r_2$ for $K_2$, where $r_1 \approx 0.19$ is the SRG separation in Fig.~\ref{fig:pendulum_srgs} for $K_1$, and $r_2 \approx 0.81$ for $K_2$.

\begin{figure}[tb]
    \centering
    \begin{subfigure}[b]{\linewidth}
        \centering
        \includegraphics[width=\linewidth]{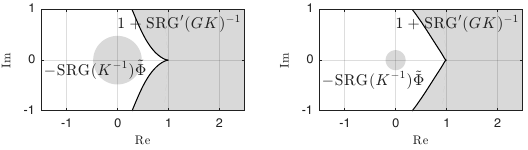}
        \caption{SRGs for $K_1$ (left) and $K_2$ (right).}
        \label{fig:pendulum_srgs}
    \end{subfigure}
    \hfill
    \begin{subfigure}[b]{\linewidth}
        \centering
        \includegraphics[width=\linewidth]{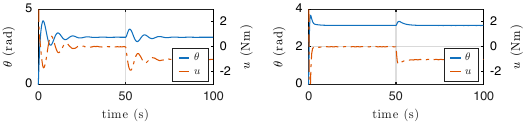}
        \caption{Simulation for $K_1$ (left) and $K_2$ (right).}
        \label{fig:pendulum_simulations}
    \end{subfigure}
    \caption{Analysis of the nonlinear pendulum example.}
    \label{fig:pendulum}
\end{figure}

We will briefly discuss why Theorem~\ref{thm:chaffey_thm2} cannot be applied in this case. When representing the controlled pendulum as $H_1 = L$ and $ H_2 = I$, with $L = (G^{-1} + \phi)^{-1} K$, one encounters the problem that $L$ is not a bounded operator, since both $K$ and the open-loop pendulum $(G^{-1} + \phi)^{-1}$ are not bounded. The same problem occurs with $H_1 = L = ((GK)^{-1} + K^{-1} \phi)^{-1}$ and $H_2=I$, where additionally it is evident that one cannot find a bounded SRG for $L$ using currently available SRG bounds, since from Fig.~\ref{fig:pendulum_srgs}, is clear that $-\SRG(K^{-1}) \Phi \cap \mathrm{Nyq}((GK)^{-1}) \neq \emptyset$, which means that $0 \in \SRG(K^{-1}) \Phi + \SRG(GK)^{-1} \supseteq \SRG(L)^{-1}$, where $\SRG(GK)$ is computed using Theorem~\ref{thm:srg_of_stabilizable_lti}. Therefore, we have the best estimate $\Gamma(L) = \infty$ with Theorem~\ref{thm:chaffey_thm2}, hence $H_1 = L$ cannot be shown to be a bounded operator. Finally, the choice $H_1 = GK$ and $H_2 = I+K^{-1} \phi$ is also not possible as $GK \notin RH_\infty$. 

\thmspace
\begin{remark}
    By picking $\kappa=0$ for the inflation of $\Phi \supseteq \SRG(x \mapsto \sin(x))$, the homotopy method has the physical interpretation that is gravity turned off at $\tau=0$, and is gradually switched on for $\tau \in (0,1]$.
\end{remark}

\subsubsection{Simulation Results}

Similar to \AA{}str\"om et al.~\cite{astromDesignPIControllers1998}, we consider the setpoint change at $t=0$ from $\theta=0$ to $\theta = \pi$, and at time $t=50$, we apply an external torque $mgl$. The simulation results are given in Fig.~\ref{fig:pendulum_simulations}. We can see that the step response for $K_2$ is less oscillatory than for $K_1$, matching the smaller gain bound. In the spirit of~\cite{astromDesignPIControllers1998}, one can use our proposed methods to design controllers with a given (incremental) gain bound on the (complementary) sensitivity function.

\subsection{Duffing Oscillator}\label{sec:examples_duffing}

The Duffing oscillator is defined by the ODE
\begin{equation}\label{eq:duffing}
    \ddot y + \delta \dot y + \alpha y + \beta y^3 = u(t),
\end{equation}
where $\alpha, \beta, \delta \in \R$ are parameters and $u(t)\in \Lte$ is some input signal. It can be viewed as a mass-spring-damper system with nonlinear spring. The parameters for this system are chosen as $\alpha = -1, \, \beta=1$, and $\delta = 0.3$. The control objective we consider is to stabilize the origin with input $r=0$ and process noise $d$, see Fig.~\ref{fig:lure_disturbed}, where we control the sensitivity $PS: d \mapsto y$ with a PD controller. We assume that the disturbance obeys $\norm{d}_\infty \leq 1$.

\begin{figure}[tb]
    \centering
    \begin{subfigure}[b]{0.6\linewidth}
        \centering

        \tikzstyle{block} = [draw, rectangle, 
        minimum height=1.5em, minimum width=2em]
        \tikzstyle{sum} = [draw, circle, node distance={0.5cm and 0.5cm}]
        \tikzstyle{input} = [coordinate]
        \tikzstyle{output} = [coordinate]
        \tikzstyle{pinstyle} = [pin edge={to-,thin,black}]
        
        \begin{tikzpicture}[auto, node distance = {0.2cm and 0.5cm},scale=0.8,transform shape]
            \node [input, name=input] {};
            \node [sum, right = of input] (sum) {$ $};
            \node [block, right = of sum] (controller) {$K(s)$};
            \node [sum, right = of controller] (sigma) {$ $};
            \node [input, above = of sigma] (disturbance) {};
            \node [block, right = of sigma] (lti) {$G(s)$};
            \node [coordinate, right = of lti] (z_intersection) {};
            \node [output, right = of z_intersection] (output) {}; %
            \node [block, below = of lti] (static_nl) {$\phi$};
            \node [coordinate, right = of static_nl] (phi_intersection) {};
        
            \draw [->] (input) -- node {$r$} (sum);
            \draw [->] (sum) -- node {$e$} (controller);
            \draw [->] (controller) -- node {$u$} (sigma);
            \draw [->] (sigma) -- node {$u'$} (lti);
            \draw [->] (lti) -- node [name=z] {$y$} (output);
            \draw [->] (z) |- (static_nl);
            \draw [->] (static_nl) -| node[pos=0.99] {$-$} (sigma);
            \node [coordinate, below = of static_nl] (tmp1) {$H(s)$};
            \draw [->] (z) |- (tmp1)-| node[pos=0.99] {$-$} (sum);
            \draw [->] (disturbance) -- node [pos=-0,above=0mm]{$d$} (sigma);
        
        \end{tikzpicture}
        \caption{Block diagram of the controlled Duffing oscillator.}
        \label{fig:lure_disturbed}
    \end{subfigure}
    \hfill
    \begin{subfigure}[b]{0.39\linewidth}
        \centering
        \includegraphics[width=\linewidth]{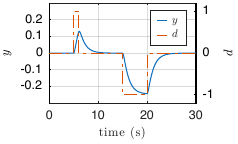}
        \caption{Simulation of response to disturbance.}
        \label{fig:sim_duffing_pd}
    \end{subfigure}
    \hfill
    \begin{subfigure}[b]{\linewidth}
        \centering
        \includegraphics[width=\linewidth]{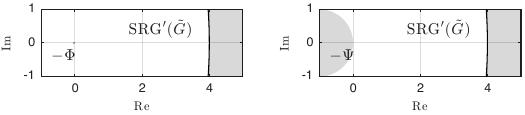}
        \caption{SRGs with $\norm{y}_\infty \leq 0.25$ (left) and $\norm{y}_\infty \leq \sqrt{2}$ (right).}
        \label{fig:srgs_duffing_pd}
    \end{subfigure}
    \caption{Analysis of the controlled Duffing oscillator.}
    \label{fig:duffing_pd_srg_sim}
\end{figure}

\subsubsection{Stability and Performance Analysis with SRGs}

In order to write the map $PS$ in the form of Fig.~\ref{fig:lure_disturbed}, we rewrite~\eqref{eq:duffing} as a Lur'e system with $G(s) = 1/(s^2 + \delta s + \alpha)$ and $\phi(y) = \beta y^3$ with $u'=u-\phi(y)+d$ in Fig.~\ref{fig:lure_disturbed}. The PD controller is given in the form of $K(s)=k_p+k_d s/(s/N+1)$. With $r=0$, Fig.~\ref{fig:lure_disturbed} can be cast in Lur'e form by setting $r = d$ in Fig.~\ref{fig:lure}, and replacing $G$ with $\tilde{G} = G/(1+GK)$, which can be verified from the relation $y = (G^{-1} + K +\phi)^{-1} d$ and~\eqref{eq:lure_closedloop}. 

The nonlinearity $x \mapsto x^3$ is problematic, as it is not bounded, see Appendix~\ref{app:srg_cubic_nl}. To remedy this, we consider an upper bound $\norm{y}_\infty \leq M$ under the input condition $\norm{d}_\infty \leq 1$, which can be obtained for the Duffing oscillator based on the approach detailed in Appendix~\ref{app:duffing_bound}. Since we study the behavior w.r.t. the zero solution $r=d=0$, we can consider $L_2$-gain and use the SG around zero. For values $k_p=k_d=5$ and $N \to \infty$, the bound in Appendix~\ref{app:duffing_bound} reads $\norm{y}_\infty \leq 0.25 =: A$. Using Theorem~\ref{thm:cubic_nl}, we obtain the bound $\SG_{\mathcal{U}_A, 0}(\phi) \subseteq D_{[0,A^2]} = : \Phi$.

Now, we apply Theorem~\ref{thm:srg_nyquist_lure} for the Lur'e system and compute~\eqref{eq:thm_srg_nyquist_lure},
which gives $r=4$. See Fig.~\ref{fig:srgs_duffing_pd} for the SRG analysis, where $N=100$ is used for the PD controller. Assuming $PS$ is well-posed, we can conclude $PS: \Lte \to \Lte$ with $\gamma(PS) \leq 1/4$.

Instead of estimating an upper bound for $\norm{y}_\infty$, one can design the controller for a predefined operating range $y \in [-\sqrt{2}, \sqrt{2}]$, as done in~\cite{koelewijnEquilibriumIndependentControlContinuousTime2023}. This allows one to use $\SG_0(\phi) \subseteq D_{[0, 2 \beta]} =: \Psi$, see Appendix~\ref{app:srg_cubic_nl}, directly in SRG calculations. The latter bound is used in Fig.~\ref{fig:srgs_duffing_pd}, where a bound $\gamma(PS) \leq 1/4$ is obtained. The result is analogous to the result using $\norm{y}_\infty \leq 0.25$, which means that, in this case, the resulting $L_2$-gain is not critically dependent on the amplitude bound $\norm{y}_\infty$.

\subsubsection{Simulation Results}

The disturbed Duffing oscillator is simulated on the time interval $t \in [0, 30]$ for the parameters $\alpha = -1, \, \beta =1, \, \delta = 0.3$ and controller parameters $k_p=k_d=5, \, N=100$. At $t=5$, we apply $d=1$ for one second and at $t=15$ we apply $d=-1$ for five seconds. The results are plotted in Fig.~\ref{fig:sim_duffing_pd}. We see that indeed the response of the system is at most $1/4$ the size of the disturbance (see $t=20$), which corresponds to $\gamma(PS) \leq 1/4$.

\subsection{Controlled Lur'e Plant with Saturation}\label{sec:example_lure_sat}

As a last example, we consider a controlled Lur'e plant with an input saturation, as depicted in Fig.~\ref{fig:controlled_lure_sat}. The LTI blocks are given by
\begin{equation*}
    K(s) = \frac{1}{s+1}, \quad G(s) = \frac{3}{(s-2)(s/10+1)},
\end{equation*} 
and the nonlinearities are the static functions
\begin{equation*}
    \phi_1(x) = 
    \begin{cases}
        x  \text{ if } |x| \leq 1, \\ 
        x/|x|  \text{ else}, 
    \end{cases} 
    \phi_2(x) = 
    \begin{cases}
        x \text{ if } |x| \leq 1, \\ 
        2x - x/|x| \text{ else}.
    \end{cases}
\end{equation*}

\begin{figure}[tb]
    \centering
    \begin{subfigure}[b]{0.9\linewidth}
        \centering
        \tikzstyle{block} = [draw, rectangle, 
        minimum height=1.5em, minimum width=2em]
        \tikzstyle{sum} = [draw, circle, node distance={0.5cm and 0.5cm}]
        \tikzstyle{input} = [coordinate]
        \tikzstyle{output} = [coordinate]
        \tikzstyle{pinstyle} = [pin edge={to-,thin,black}]
        
        \begin{tikzpicture}[auto, node distance = {0.2cm and 0.5cm},scale=0.8,transform shape]
            \node [input, name=input] {};
            \node [sum, right = of input] (sum) {$ $};
            \node [block, right = of sum] (controller) {$K(s)$};
            \node [block, right = of controller] (saturation) {$\phi_1$};
            \node [sum, right = of saturation] (sigma) {$ $};
            \node [block, right = of sigma] (lti) {$G(s)$};
            \node [coordinate, right = of lti] (z_intersection) {};
            \node [output, right = of z_intersection] (output) {}; %
            \node [block, below = of lti] (static_nl) {$\phi_2$};
            \node [coordinate, right = of static_nl] (phi_intersection) {};
        
            \draw [->] (input) -- node {$r$} (sum);
            \draw [->] (sum) -- node {$e$} (controller);
            \draw [->] (controller) -- node {$u$} (saturation);
            \draw [->] (saturation) -- node {$\hat{u}$} (sigma);
            \draw [->] (sigma) -- node {$u'$} (lti);
            \draw [->] (lti) -- node [name=z] {$y$} (output);
            \draw [->] (z) |- (static_nl);
            \draw [->] (static_nl) -| node[pos=0.99] {$-$} (sigma);
            \node [coordinate, below = of static_nl] (tmp1) {$H(s)$};
            \draw [->] (z) |- (tmp1)-| node[pos=0.99] {$-$} (sum);
        
        \end{tikzpicture}
        \caption{A controlled Lur'e plant with input saturation.}
        \label{fig:controlled_lure_sat}
    \end{subfigure}
    \hfill
    \begin{subfigure}[b]{\linewidth}
        \centering
        \includegraphics[width=\linewidth]{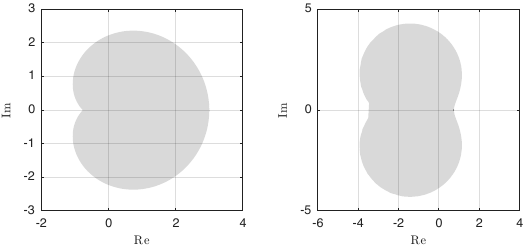}
        \caption{SRG bound $\mathcal{C}(L)$ for $L$ (left) and $\mathcal{C}(T)$ for $T$ (right).}
        \label{fig:lure_sat_srgs}
    \end{subfigure}
\caption{Analysis of the example in Section~\ref{sec:example_lure_sat}.}
\label{fig:lure_sat_srg}
\end{figure}

Since $\phi_1, \phi_2$ are static nonlinearities, bounds for their SRGs are obtained via Proposition~\ref{prop:static_nl_srg} as $\SRG(\phi_1) \subseteq D_{[0,1]} =: \Phi$ and $\SRG(\phi_2) \subseteq D_{[1, 2]} = :\Psi$. Since $\Phi$ and $\Psi$ are disks, they are inflatable and obey the chord and arc properties.

The interconnection in Fig.~\ref{fig:controlled_lure_sat}, which defines the map $T : r \mapsto y$, can be studied by combining Theorems~\ref{thm:chaffey_thm2} and~\ref{thm:srg_nyquist_lure}. We represent the closed-loop relation as $T = (I+L^{-1})^{-1}$ with $L = (G^{-1} + \phi_2)^{-1} \phi_1 K$. The SRG analysis consists of three steps. First, we analyze $(G^{-1} + \phi_2)^{-1}$ (i.e., the map $\hat{u} \mapsto y$) using Theorem~\ref{thm:srg_nyquist_lure}. This step is worked out in detail in~\cite{krebbekxScaledRelativeGraph2025a}. Second, we use Theorem~\ref{thm:srg_of_stabilizable_lti} for $\SRG(G)$ and Proposition~\ref{prop:srg_calculus} to compute an SRG bound 
\begin{equation*}
    \mathcal{C}(L) := (\SRG(G)^{-1}+\Psi)^{-1}\Phi \SRG(K),
\end{equation*}
such that $\SRG(L) \subseteq \mathcal{C}(L)$, see Fig.~\ref{fig:lure_sat_srgs}. Well-posedness and $L_2$-stability of $L$ follows from the application of Theorem~\ref{thm:srg_nyquist_lure} to $(G^{-1} + \phi_2)^{-1}$, the fact that $K$ and $\phi_1$ have finite incremental $L_2$-gain, and $\Gamma(L) \leq \rmin(\mathcal{C}(L))$. Finally, with this bound for $\SRG(L)$, we use Theorem~\ref{thm:chaffey_thm2} to analyze the well-posedness of $T_\mathrm{e}$, stability and incremental $L_2$-gain of $T$, as well as the bound $\SRG(T) \subseteq \mathcal{C}(T)$ in Fig.~\ref{fig:lure_sat_srgs}, computed as $\mathcal{C}(T) := (1+\mathcal{C}(L)^{-1})^{-1}$, from which we can conclude $\Gamma(T) \leq \rmin(\mathcal{C}(T)) \approx 4.81$.

\section{Conclusion}\label{sec:conclusion}

In this paper, we extend the SRG framework to enable the treatment of unstable LTI loop components. The key novelty is an extended SRG for LTI operators, where the encirclement information from the Nyquist criterion is incorporated into the SRG. The extended SRG resolves a pitfall of the state-of-the-art SRG methods when applied to unstable systems, which we identified in this paper. We show how stability and $H_\infty$-performance can be derived from the extended SRG, and prove that the SRG of unstable LTI systems that can be stabilized by static feedback is precisely the hyperbolic convex hull of their Nyquist diagram, analogous existing results for the stable case. 

The results are applied to three feedback systems: the Lur'e system, the controlled Lur'e system, and the LTI plant with Lur'e controller, for which we derive theorems that guarantee well-posedness, stability and (incremental) $L_2$-gain performance in terms of a graphical condition involving the extended SRG. The result for the Lur'e system yields a generalization of the celebrated circle criterion, where an upper bound for the $L_2$-gain is derived from the separation between the Nyquist curve and the critical disk from the classical circle criterion.

Finally, we showcase the use of our methods in three examples, which include nonlinearities and unstable LTI plants, where existing SRG analysis methods are not applicable.

\begin{ack}                               
    We thank the anonymous reviewers of an earlier draft for their valuable feedback, which significantly improved the quality of the manuscript
\end{ack}

\appendix

\section{Relations in Feedback Interconnection}\label{app:relations_feedback}

The map from $r$ to $y$ in Fig.~\ref{fig:chaffey_thm2} can be written as a relation in more than one way. One could start with $y = H_1 e$ and invert $e = r-H_2 H_1 e$ to obtain $e = (I+H_2 H_1)^{-1} r$, hence $y = H_1 (I+H_2 H_1)^{-1} r$. Alternatively, as in~\cite{chaffeyGraphicalNonlinearSystem2023}, one can invert $H_1$ as a relation to obtain $e = H_1^{-1} y$ and invert $e = H_1^{-1} y = r-H_2 y$ to arrive at $y = (H_1^{-1} + H_2)^{-1}r$. The next lemma shows how the relations $(H_1^{-1} + H_2)^{-1}$ and $H_1 (I+H_2 H_1)^{-1}$ are connected.

We recall the definition of the inverse, sum and product of relations~\cite{chaffeyGraphicalNonlinearSystem2023,bauschkeConvexAnalysisMonotone2017}:
\begin{subequations}\label{eq:relations}
\begin{align}
    R^{-1} &= \{ (y,x) \mid (x,y) \in R \}, \label{eq:relations-inverse}\\
    R+S &= \{ (x,y+z) \mid (x,y) \in R, (x,z) \in S \}, \label{eq:relations-sum}\\
    T R &= \{ (x,z) \mid \exists y:(x,y) \in R, \, (y,z)\in T \}. \label{eq:relations-product}
\end{align}
\end{subequations} %

\thmspace
\begin{lemma}\label{lemma:relational-equivalance}
    Let $H_1 : X \rightrightarrows Y$ and $H_2 : Y \rightrightarrows X$ be relations on vector spaces $X,Y$, then $(H_1^{-1} + H_2)^{-1} \subseteq H_1 (I+H_2 H_1)^{-1}$.  Equality holds if $H_1$ is single-valued.
\end{lemma}
\begin{proof}
    First, note that 
    \begin{align}
        &(H_1^{-1} + H_2)^{-1} = 
        \{ (u, y) \mid (y,u) \in H_1^{-1} + H_2 \}=  \nonumber \\
        &\{ (u_1 + u_2, y) \mid (y, u_1) \in H_1^{-1} , \, (y, u_2) \in H_2 \} =\nonumber \\
        & \{ (u_1 + u_2, y) \mid (u_1, y) \in H_1 , \, (y, u_2) \in H_2  \}.  \label{eq:relation-srg-expression}
    \end{align}
    Similarly, 
    \begin{align}
        &H_1 (I + H_2 H_1)^{-1} = \{ (u, y) \mid y \in H_1 (I + H_2 H_1)^{-1} u \} =\nonumber \\
        & \{ (u, y) \mid \exists e \text{ s.t. } (e,y) \in H_1, \, (u, e) \in (I+H_2 H_1)^{-1} \}= \nonumber \\
        & \{ (u, y) \mid \exists e \text{ s.t. } (e,y) \in H_1, \, (e, u) \in I+H_2 H_1 \}= \nonumber \\
        & \{ (u_1 + u_2, y) \mid (u_1,y) \in H_1, \, (u_1, u_2) \in H_2 H_1 \}, \label{eq:relation-operator-expression}
    \end{align}
    where we used that $(e,u_1) \in I \iff e=u_1$. Observe that $(u_1,u_2) \in H_2 H_1$ follows from~\eqref{eq:relation-srg-expression} and~\eqref{eq:relations-product}, which proves $(H_1^{-1}+H_2)^{-1} \subseteq H_1(I+H_2 H_1)^{-1}$. 

    Note that~\eqref{eq:relation-operator-expression} implies that $(u_1,y) \in H_1$ and, via~\eqref{eq:relations-product}, that $(u_1,\tilde{y}) \in H_1$ and $(\tilde{y},u_2) \in H_2$ for some $\tilde{y} \in Y$. When $H_1$ is single-valued, $\tilde{y}=y$ must hold, hence $(u_1,y) \in H_1$ and $(y,u_2) \in H_2$, proving $(H_1^{-1}+H_2)^{-1} \supseteq H_1(I+H_2 H_1)^{-1}$. 
\end{proof}

An example where equality is violated is $H_1 = \R \times \R$ and $H_2 = \{(x,0) \mid 0 \neq x \in \R \}$, where $(0,0) \in H_1 (I+H_2 H_1)^{-1}$ holds, while $(0,0) \notin (H_1^{-1}+H_2)^{-1}$.

\section{The Nyquist and Circle Criteria}\label{app:nyquist_circle}

The Nyquist criterion can be used to assess the stability of the LTI feedback system in Fig.~\ref{fig:linear_feedback}.%

\thmspace
\begin{theorem}[\hspace*{-3pt}{\cite[Thm. IV.5.5]{desoerFeedbackSystemsInputoutput1975}}]\label{thm:nyquist}
    For proper $L \in \R(s)$, let $n_\mathrm{z}$ denote the number of unstable closed-loop poles, $n_\mathrm{p}$ the number of unstable poles of $L(s)$\footnote{Poles of $L(s)$ on $j \R$ do not contribute to $n_\mathrm{p}$. The D-contour is chosen such that poles on $j \R$ are kept to the left of the D-contour by using infinitesimally small indentations.}, and $n_\mathrm{n}$ the number of clockwise encirclements of $-1$ by $L(j \omega)$ as $\omega $ traverses the $D$-contour, going from $-jR$ to $jR$ and then along $Re^{j \phi}$ as $\phi$ goes from $\pi/2 \to -\pi/2$, for $R \to \infty$. The closed-loop system in Fig.~\ref{fig:linear_feedback} satisfies $n_\mathrm{z}=n_\mathrm{n}+n_\mathrm{p}$.
\end{theorem}
\thmspace

\noindent Note that for stability, $n_\mathrm{z}=0$ is required. Stability of a Lur'e system can be assessed using the circle criterion.%

\thmspace
\begin{theorem}[\hspace*{-3pt}{\cite[Thm. V.2.10]{desoerFeedbackSystemsInputoutput1975}}]\label{thm:circle}
    Let $G \in \R(s)$ be proper, and let $k_1 \leq \phi(x)/x \leq k_2, \ \forall x\in \R$. Let $n_\mathrm{p}$ be the number of poles $p$ of $G(s)$ such that $\mathrm{Re} (p)>0$. If the system defined by $y_1 = G e_1$, $e_1 = u_1 - y_2$, $y_2 = \phi(e_2)$ and $e_2 = u_2 + y_1$ is well-posed, then $u_1,u_2 \in L_2 \Rightarrow e_1,e_2,y_1,y_2 \in L_2$ when:
    \begin{enumerate}
        \item if $0<k_1<k_2$, then the Nyquist diagram of $G(s)$ must not intersect $D_{[-1/k_1, -1/k_2]}$ and has to encircle it $n_p$ times in counterclockwise direction.
        \item if $0=k_1<k_2$, then $n_\mathrm{p}=0$ must hold and the Nyquist diagram must satisfy 
        \begin{equation*}
            \mathrm{Re} \, G(j\omega) > -1/k_2, \quad \forall \omega \in \R.
        \end{equation*}
        \item if $k_1 <0 <k_2$, then $n_\mathrm{p}=0$ must hold and the Nyquist diagram has to be contained in the interior of $D_{[-1/k_1, -1/k_2]}$.
    \end{enumerate}
\end{theorem}
\thmspace

\section{Complex Geometry}\label{app:complex_geometry}

In this appendix, we recall the chord and arc definitions from~\cite{ryuScaledRelativeGraphs2022}. Moreover, we extend the hyperbolic convexity notion from~\cite{huangScaledRelativeGraph2020,patesScaledRelativeGraph2021} to $\C_\infty$, and develop the necessary complex geometry facts to prove Lemma~\ref{lemma:srg_radius_gamma_Gamma}, Theorem~\ref{thm:srg_of_stabilizable_lti}, and Corollary~\ref{cor:generalized_circle}.

We denote the line segment between $z_1, z_2 \in \C$ as $[z_1, z_2] := \{ \alpha z_1 + (1-\alpha) z_2 \mid \alpha \in [0, 1] \}$. The right-hand arc and left-hand arc are defined as
\begin{align*}
    \operatorname{Arc}^+(z, \bar{z}) &= \\ \{ r &e^{j(1-2 \alpha)\phi}  \, \mid z = r e^{j \phi}, \phi \in (-\pi, \pi ], \alpha \in [0, 1] \}, \\
    \operatorname{Arc}^-(z, \bar{z}) &= -\operatorname{Arc}^+(-z, -\bar{z}),
\end{align*}
respectively~\cite{ryuScaledRelativeGraphs2022}. For $z_1, z_2 \in \C_{\mathrm{Im} \geq 0}$, where we assume w.l.o.g. that $\mathrm{Re} (z_1) \leq \mathrm{Re} (z_2)$, $\operatorname{Circ}(z_1, z_2)$ is the unique circle through $z_1, z_2$ centered on $\R$, which is $\mathrm{Re}(z_1)+j \R$ if $\mathrm{Re}(z_1)=\mathrm{Re}(z_2)$ or $z_2 = \infty$, and let~\cite{huangScaledRelativeGraph2020}
\begin{multline*}
    \operatorname{Arc}_{\operatorname{min}}(z_1, z_2) =  \{ z \in \operatorname{Circ}(z_1, z_2)  \mid \\ \mathrm{Re} (z_1) \leq \mathrm{Re} (z) \leq \mathrm{Re} (z_2), \, \mathrm{Im} (z) \geq 0  \},
\end{multline*} 
if $\mathrm{Re}(z_1) < \mathrm{Re}(z_2)$, else $\operatorname{Arc}_{\operatorname{min}}(z_1, z_2) = [z_1,z_2]$ and $[z,\infty] : = \{w \in \C \mid w = z + jx, \, x \in \R_{\geq 0} \} \cup \{ \infty\}$. Furthermore, $\mathrm{Arc}_\mathrm{min}$ has the following properties.

\thmspace
\begin{lemma}\label{lemma:armin_abs}
    For distinct $z_1,z_2 \in \C_{\mathrm{Im} \geq 0} \cup \{ \infty \}$ and $A := \{|z| \mid z \in \mathrm{Arc}_\mathrm{min}(z_1,z_2)\}$, one has $\inf A = \min\{|z_1|,|z_2|\}$ and $\sup A = \max\{|z_1|,|z_2|\}$.
\end{lemma}
\begin{proof}
    If $z_2 = \infty$ then the statement follows from $\mathrm{Arc}_\mathrm{min}(z_1,\infty) = [z_1,\infty]$. If $z_1,z_2 \in \C_{\mathrm{Im} \geq 0}$, then $\mathrm{Circ}(z_1,z_2)$ is parametrized by $\phi \in [0,\pi]$ as $z(\phi)=c+re^{j \phi}$ for some $c,r\in \R$. Since $|z(\phi)|^2=c^2+r^2+2cr \cos(\phi)$, $|z(\phi)|$ is a monotonically increasing or decreasing function of $\phi$, proving the statement.
\end{proof}

\thmspace
\begin{lemma}\label{lemma:armin_re}
    For $z_1,z_2 \in \C_{\mathrm{Im} \geq 0}$ with $\mathrm{Re}(z_1)<\mathrm{Re}(z_2)$ and $R := \{\mathrm{Re}(z) \mid z \in \mathrm{Arc}_\mathrm{min}(z_1,z_2)\}$, one has $\inf R = \mathrm{Re}(z_1)$ and $\sup R = \mathrm{Re}(z_2)$.
\end{lemma}
\begin{proof}
    By the same argument as in the proof of Lemma~\ref{lemma:armin_abs}, one has $\mathrm{Re}(z(\phi)) = c+r \cos(\phi)$, which is monotone in $\phi$, proving the statement.
\end{proof}

\noindent We define hyperbolic convexity as in~\cite{chaffeyGraphicalNonlinearSystem2023,huangScaledRelativeGraph2020}, including $\infty$.

\thmspace
\begin{definition}\label{def:hco}
    A set $S \subseteq \C_{\mathrm{Im} \geq 0} \cup \{ \infty \}$ is h-convex if 
    \begin{equation*}
        z_1, z_2 \in S \implies \operatorname{Arc}_{\operatorname{min}}(z_1, z_2) \subseteq S.
    \end{equation*}
    The h-convex hull of $P \subseteq \C_{\mathrm{Im} \geq 0} \cup \{ \infty \}$, is the smallest h-convex set containing $P$, denoted as $\operatorname{hco}(P)$. 
\end{definition}
\thmspace

For a set $P \subseteq \C_\infty$ that obeys $\bar{P} = P$, h-convexity can be studied for $P_+ := P \cap (\C_{\mathrm{Im}  \geq 0} \cup \{ \infty \})$. The h-convex hull is defined as $\hco(P) := \operatorname{hco}(P_+) \cup \overline{\operatorname{hco}(P_+)}$, which is identical to~\cite{patesScaledRelativeGraph2021} when $\infty$ is not considered.

\thmspace
\begin{lemma}\label{lemma:hco_no_inf}
    If a set $S \subseteq \C_{\mathrm{Im} \geq 0}$ is h-convex, then so are $\alpha S$ for $\alpha \in \R$, $1+S$ and the M\"obius inverse $(S \setminus \{ 0 \})^{-1}$.
\end{lemma}
\begin{proof}
    H-convexity of $\alpha S$ and $1+S$ follows from noting $\operatorname{Arc}_{\operatorname{min}}(\alpha z_1, \alpha z_2) = \alpha \operatorname{Arc}_{\operatorname{min}}(z_1, z_2)$ and $\operatorname{Arc}_{\operatorname{min}}(1+ z_1, 1+ z_2) = 1+ \operatorname{Arc}_{\operatorname{min}}( z_1, z_2)$.

    To prove h-convexity of $S^{-1}$, it suffices to prove $\operatorname{Arc}_{\operatorname{min}}(\bar z_1^{-1}, \bar z_2^{-1}) = \operatorname{Arc}_{\operatorname{min}}(z_1, z_2)^{-1}$. First, we will show $\operatorname{Circ}(\bar z_1^{-1}, \bar z_2^{-1}) = \operatorname{Circ}(z_1, z_2)^{-1}$. For each $z_1,z_2 \in S$, we are in precisely one of the following cases:
    \begin{itemize}
        \item [1.] w.l.o.g. $\mathrm{Re}(z_1) < \mathrm{Re}(z_2)$. Here, $\operatorname{Circ}(z_1, z_2)$ is a circle in $\C$ with center $c \in \R$ and radius $R$, defined as $z\in \C$ such that $(z-c)(\bar{z}-c)=R^2$ (or $|z-c|=R$).
        \begin{itemize}
            \item [1.a.] $c \neq \pm R$. Inserting $z = \bar w^{-1}$ yields $(\bar w^{-1}-c)(w^{-1}-c)=R^2 \Leftrightarrow 1-c \bar w-c w+ w \bar w (c^2-R^2)= 0 \Leftrightarrow (w-\frac{c}{c^2-R^2})(\bar w - \frac{c}{c^2-R^2})=\frac{R^2}{(c^2-R^2)^2}$.
            \item [1.b.] $c = \pm R$: The calculation in the previous case yields $1-c w-c \bar w=0$ $\Leftrightarrow$ $\mathrm{Re}(w) = \frac{1}{2c}$.
        \end{itemize}
        \item [2.] $\mathrm{Re}(z_1) = \mathrm{Re}(z_2)$. Here, $\operatorname{Circ}(z_1, z_2)$ is defined as $z \in \C$ such that $z+\bar z = 2 x_0$, where $x_0 = \mathrm{Re}(z_1)$.
        \begin{itemize}
            \item [2.a.] $x_0 \neq 0$. Inserting $z = \bar w^{-1}$ yields $w+\bar w = 2 x_0 w \bar w \Leftrightarrow (w-\frac{1}{2 x_0})(\bar w- \frac{1}{2 x_0})=\frac{1}{4 x_0^2}$.
            \item [2.b.] $x_0 = 0$. Note $(j\R \setminus \{ 0 \})^{-1} = j\R \setminus \{ 0 \}$. 
        \end{itemize}
    \end{itemize}
    As $\operatorname{Circ}(z_1, z_2)$ is uniquely determined by $z_1,z_2$, and $\bar z_1^{-1}, \bar z_2^{-1} \in \operatorname{Circ}(z_1, z_2)^{-1}$ by definition, we have shown $\operatorname{Circ}(\bar z_1^{-1}, \bar z_2^{-1}) = \operatorname{Circ}(z_1, z_2)^{-1}$. 
    
    Finally, since the M\"obius inverse $z \mapsto \bar z^{-1}$ is a homeomorphism on $\C \setminus \{ 0 \}$, it maps connected sets to connected sets, hence it follows that $\operatorname{Arc}_{\operatorname{min}}(z_1, z_2) \mapsto \operatorname{Arc}_{\operatorname{min}}(z_1, z_2)^{-1} = \operatorname{Arc}_{\operatorname{min}}(\bar z_1^{-1}, \bar z_2^{-1})$.
\end{proof}

\begin{lemma}\label{lemma:hco_w_inf}
    If $S \subseteq \C_{\mathrm{Im} \geq 0} \cup \{ \infty \}$ is h-convex, then so are $\alpha S$ for $0 \neq \alpha \in \R$, $1+S$ and the M\"obius inverse $S^{-1}$.
\end{lemma}
\begin{proof}
    For $\alpha S$, note that $\operatorname{Arc}_{\operatorname{min}}(\alpha z, \alpha \infty) = [\alpha z,\infty] = \alpha [z,\infty] = \alpha  \operatorname{Arc}_{\operatorname{min}}(z, \infty)$ where $z \in \C_{\mathrm{Im} \geq 0}$ and $\alpha \infty = \infty$ for $0 \neq \alpha$~\cite{ryuScaledRelativeGraphs2022}. Similarly, for $1+S$ observe $\operatorname{Arc}_{\operatorname{min}}(1+z, 1+ \infty) = [1+z,\infty] = 1+[z,\infty] = 1+  \operatorname{Arc}_{\operatorname{min}}(z, \infty)$, where $1+\infty=\infty$~\cite{ryuScaledRelativeGraphs2022}. The remaining cases follow from Lemma~\ref{lemma:hco_no_inf}.
    
    For $S^{-1}$, note that $(\operatorname{Arc}_{\operatorname{min}}(z_1,z_2) \setminus \{0,\infty\})^{-1} = \operatorname{Arc}_{\operatorname{min}}(\bar z_1^{-1},\bar z_2^{-1}) \setminus \{0,\infty\}$ follows from Lemma~\ref{lemma:hco_no_inf}. Defining $\bar 0^{-1} = \infty$ and $\bar \infty^{-1} = 0$ yields $\operatorname{Arc}_{\operatorname{min}}(z_1,z_2)^{-1} = \operatorname{Arc}_{\operatorname{min}}(\bar z_1^{-1},\bar z_2^{-1}) $, proving the result. 
\end{proof}

\section{SRG of a Cubic Nonlinearity}\label{app:srg_cubic_nl}

The function $\phi(x)=x^3$ defines a static nonlinear map $u(t) \mapsto u^3(t)$ on signals $u : \R \to \R$. However, $\phi: \dom(\phi) \subsetneq L_2 \to L_2$. This can be shown using $u_\mathrm{c}(t) := t^{-1/4} $ for $t \in (0,1)$, else $u_\mathrm{c}(t) :=0$, which obeys $u_\mathrm{c} \in  L_2$, but $u_\mathrm{c}^3 \notin L_2$. Therefore, $\phi_\mathrm{e} : \dom(\phi_\mathrm{e}) \subsetneq \Lte \to \Lte$.

Therefore, we analyze the SRG of $\phi$ on certain subsets of $\dom(\phi)$. Particularly useful subsets of $\dom(\phi)$ are those with finite amplitude (note that they exclude $u_\mathrm{c}$).

\thmspace
\begin{theorem}\label{thm:cubic_nl}
    For $A \geq 0$, let $\mathcal{U}_A = \{ x \in L_2 \mid \norm{x}_\infty \leq A \}$. The operator $\phi : \dom(\phi) \subsetneq L_2 \to L_2$, given by $x(t) \mapsto x^3(t)$, satisfies $\mathcal{U}_A \subseteq \dom(\phi)$ and 
    \begin{enumerate}[label=\alph*.]
        \item $\SG_{\mathcal{U}_A, 0}(\phi) \subseteq D_{[0,A^2]}$,
        \item $\SRG_{\mathcal{U}_A}(\phi) \subseteq D_{[0,3A^2]}$,
        \item $\lim_{A \to \infty} \SRG_{\mathcal{U}_A}(\phi) \subseteq \C_{\mathrm{Re} > 0} \cup \{ 0 \}$.
    \end{enumerate}
\end{theorem}
\thmspace

\begin{proof}
    Given an $A \geq 0$, define $\phi_A(x) := x^3$ for $|x| \leq A$, else $\phi_A(x) := \operatorname{sign}(x) A^3$. Note that this map satisfies $\phi_A(0)=0$, $\phi_A \in [0,A^2]$ and $\partial \phi_A \in [0, 3A^2]$, hence $\phi_A : L_2 \to L_2$. Propositions~\ref{prop:static_nl_srg_non_incremental} and~\ref{prop:static_nl_srg} imply $\SG_{0}(\phi_A) \subseteq D_{[0,A^2]}$ and $\SRG(\phi_A) \subseteq D_{[0,3A^2]}$, respectively. Note that this also implies $\phi_A : L_2 \to L_2$, and since $\phi_A |_{\mathcal{U}_A} = \phi |_{\mathcal{U}_A}$, we conclude that $\phi : \mathcal{U}_A \to L_2$, i.e., $\mathcal{U}_A \subseteq \dom(\phi)$. 
    
    Parts $a.$ and $b.$ follow from the observations that $\SG_{\mathcal{U}_A,0}(\phi) = \SG_{\mathcal{U}_A,0}(\phi_A) \subseteq \SG_0(\phi_A)$ and $\SRG_{\mathcal{U}_A}(\phi) = \SRG_{\mathcal{U}_A}(\phi_A) \subseteq \SRG(\phi_A)$. Part $c.$ follows from $\lim_{A \to \infty} D_{[0,3A^2]} = \bigcup_{A>0} D_{[0,A]} =\C_{\mathrm{Re} > 0} \cup \{ 0 \}$. 
\end{proof}

\section{Amplitude Bound for the Duffing Oscillator}\label{app:duffing_bound}

In this appendix, we show how to obtain a bound $\norm{y}_\infty$ for the system in Fig.~\ref{fig:lure_disturbed}. The system with a pure PD controller $K = k_p + k_d s$ can be written as $\ddot y =- \tilde{\alpha} y - \beta y^3 - \tilde{\delta} \dot y + d$, where $\tilde{\alpha} = \alpha + k_p>0$ and $\tilde{\delta} = \delta + k_d>0$. Using $\tilde{\delta} = 0$ and $d=0$, we can define the Hamiltonian $H(y, \dot y) := \frac{1}{2}(\dot y)^2 + \frac{1}{2} \tilde{\alpha} y^2 + \frac{1}{4} \beta y^4$, which reproduces the Duffing equation upon the substitution $x = \dot y$ with $\dot y = \partial H/ \partial x, \ddot y = \dot x = - \partial H/\partial y $. This Hamiltonian is a convex potential, consisting of a linear and cubic spring, both with positive spring constant. Therefore, the maximum amplitude $\norm{y}_\infty$ of the system response is monotonically increasing with the energy $H$, which in turn means that disturbance $d$ that maximizes $H$ will result in the upper bound for $\norm{y}_\infty$. 

Multiplying the Duffing equation by $\dot y$ results in the equality $\dot y (\ddot y +\tilde{\alpha} y +\beta y^3 + \tilde{\delta} \dot y - d(t))=0$, which can be rewritten as $\frac{d}{d t} H(y, \dot y) = -\tilde{\delta} (\dot y)^2 + \dot y d(t)$. This means that $d(t)$ should always have the same sign as $\dot y$ in order to add as much energy to the system as possible. Therefore, we start at $t=0$ with a \emph{negative} value of $y$ with $\dot y=0$, and apply the maximum positive value of $d$. The motion starts with $\dot y >0$. We define the \emph{turning point} as the value of $y$ when $\dot y =0$ again for the first time.

If the system at $t=0$ starts at $(y(0), \dot y(0)) = (-\norm{y}_\infty, 0)$ with $d(t) = \norm{d}_\infty$ throughout, then the turning point should obey $(y, \dot y) = (\norm{y}_\infty, 0)$ in an equilibrium situation. Solving this problem numerically for the values $\alpha=-1, \, \beta =1, \delta=0.3$, $k_p=k_d=5$, and $\norm{d}_\infty = 1$, the resulting bound is $\norm{y}_\infty = 0.25$.

\bibliographystyle{plain} 
\bibliography{bibliography}

\end{document}